\documentclass[noinfoline,11pt]{imsart}

\RequirePackage{amsthm,amsmath}
\RequirePackage[numbers]{natbib}
\RequirePackage[colorlinks,citecolor=blue,urlcolor=blue]{hyperref}
\RequirePackage{graphicx}

\startlocaldefs

\RequirePackage{amssymb,bm}
\RequirePackage[mathscr]{eucal}
\DeclareFontEncoding{FMS}{}{}
\DeclareFontSubstitution{FMS}{futm}{m}{n}
\DeclareFontEncoding{FMX}{}{}
\DeclareFontSubstitution{FMX}{futm}{m}{n}
\DeclareSymbolFont{fouriersymbols}{FMS}{futm}{m}{n}
\DeclareSymbolFont{fourierlargesymbols}{FMX}{futm}{m}{n}
\DeclareMathDelimiter{\hsnorm}{\mathord}{fouriersymbols}{152}{fourierlargesymbols}{147}
\newcommand{\hs}{\big\hsnorm}

\newtheorem{theorem}{Theorem}

\newtheorem{lemma}[theorem]{Lemma}
\newcommand{\tr}{{\rm tr}}

\newcommand{\innprod}[2]{{\left\langle {#1},{#2}\right\rangle}}
\newcommand{\topt}[2]{\mathbf t_{#1}^{#2}}
\newcommand{\Tan}{\mathrm{Tan}}
\newcommand{\range}{\mathrm{range}}

\newcommand {\X}{\mathcal{H}}

\usepackage{marginnote}

\endlocaldefs

\begin{document}

\begin{frontmatter}

\title{Transportation-Based Functional ANOVA and PCA for Covariance Operators}
\runtitle{Transportation-Based Covariance ANOVA and PCA}

\begin{aug}
  \author{\fnms{Valentina} \snm{Masarotto}\ead[label=e1]{v.masarotto@math.leideuniv.nl}}
\address{Mathematisch Instituut\\ Universiteit Leiden, The Netherlands\\ \printead{e1}}

\author{\fnms{Victor M.} \snm{Panaretos}\ead[label=e2]{victor.panaretos@epfl.ch}}
\and
 \author{\fnms{Yoav} \snm{Zemel} \ead[label=e3]{yoav.zemel@epfl.ch}}

\address{Institut de Math\'ematiques\\ \'Ecole Polytechnique F\'ed\'erale de Lausanne,  Suisse \\ \printead{e2,e3}}
  

\end{aug}

\begin{abstract}
We consider the problem of comparing several samples of stochastic processes with respect to their second-order structure, and describing the main modes of variation in this second order structure, if present. These tasks can be seen as an Analysis of Variance (ANOVA) and a Principal Component Analysis (PCA) of covariance operators, respectively. They arise naturally in functional data analysis, where several populations are to be contrasted relative to the nature of their dispersion around their means, rather than relative to their means themselves.  
We contribute a novel approach based on optimal (multi)transport, where each covariance can be identified with a a centred Gaussian process of corresponding covariance. By means of constructing the optimal simultaneous coupling of these Gaussian processes, we contrast the (linear) maps that achieve it with the identity with respect to a norm-induced distance. The resulting test statistic, calibrated by permutation, is seen to distinctly outperform the state-of-the-art, and to furnish considerable power even under local alternatives. This effect is seen to be genuinely functional, and is related to the potential for perfect discrimination in infinite dimensions.  In the event of a rejection of the null hypothesis stipulating equality, a geometric interpretation of the transport maps allows us to construct a (tangent space) PCA revealing the main modes of variation. As a necessary step to developing our methodology, we prove results on the existence and boundedness of optimal multitransport maps. These are of independent interest in the theory of transport of Gaussian processes. The transportation ANOVA and PCA are illustrated on a variety of simulated and real examples.
\end{abstract}

\end{frontmatter}

\tableofcontents

\section{Introduction}
\label{sec:intro}

Let $\{X_{i,1}\}_{i=1}^{n_1}, \dots ,\{X_{i,K}\}_{i=1}^{n_K}$ be $K$ independent samples of i.i.d. random elements in a separable Hilbert space $\mathcal{H}$, posessing well-defined means $\{\mu_j\}_{j=1}^K$ and covariances $\{\Sigma_j\}_{j=1}^{K}$. We consider the problem of testing the hypothesis 
\begin{equation}\label{null_hypothesis}
H_0: \Sigma_1=\Sigma_2=\hdots=\Sigma_K
\end{equation}
on the basis of the observations $\{X_{i,j}\}$ and, if $H_0$ is rejected, the subsequent problem of describing the main mode(s) of variation of the $K$ underlying covariances.

This problem arises very naturally in functional data analysis, i.e. when $\mathcal{H}$ is taken to be a function space (for instance $L^2[0,1]$ or a reproducing kernel Hilbert subspace thereof), and one is interested in discerning whether $K$ different groups of functions manifest the same type of dispersion relative to their mean. For instance, the functions could be curves representing DNA minicircles (\citet{panaretos2010second}, \citet{kraus2012dispersion}, and \citet{tavakoli:2014}), where different groups correspond to different base-pair sequences, and one is interested in probing for a dependence of the mechanical properties on the base pair sequence; or they could be surfaces representing the log spectrograms of short spoken words by different speakers (as in \citet{ferraty:2006}), and one may wish to see whether there is a difference in several groups of sounds; yet a further example may be in the analysis of age-dependent wheel-running
activity curves in mice, where one may wish to see whether the level of activity across age had evolved under several generation selections (\citet{cabassi2017permutation}).
What is common to all these examples is that it is not the mean structure that is suspected to differ (or at least to capture the most interesting differences); in that case, the problems would fall under the well studied topic of functional analysis of variance (see, e.g., \citet{benko:2009,zhang2013analysis, cuesta2010simple, gorecki2015comparison}). Rather it is the fluctuations around the means $\mu_j$, as encapsulated by the operators $\Sigma_j$, in what could be termed a functional covariance ANOVA. 

Early contributions in this direction focus on the two-sample case, as in \citet{benko2009common, panaretos2010second, fremdt2013testing}. In particular, \citet{benko2009common} propose a two–sample bootstrap based tests for some aspects of the spectrum of functional data, \citet{panaretos2010second} consider the problem in a two-sample setting with Gaussian processes, and \citet{fremdt2013testing} extend to non-Gaussian. \citet{kraus2012dispersion} provide resistant versions of two sample tests, focussed on operators related to the covariance. Two-sample testing has been first extended to the $K$-sample case by \citet{boente2018testing}, who propose a test based on the Hilbert--Schmidt distance between the estimated covariance operators of each population and where the critical values of the test statistics are calibrated via a bootstrap procedure.  The common theme in these papers is that the covariances are contrasted with respect to the Hilbert--Schmidt metric, which corresponds to imbedding covariance operators in a larger linear space, whereas they are not closed under linear operations. Instead, covariance operators are trace-class non-negative operators, so rather than being seen as Hilbert--Schmidt operators, they are better represented as ``squares" of such operators. For this reason, \citet{pigoli2014distances} considered the use of nonlinear metrics adapted to non-negative operators in the two-sample setting, generalising some of the work in \citet{dryden2009non} in finite dimensions. \citet{cabassi2017permutation} extend their metric-based methodology to the $K$-sample case. Other contributions for the $K$-sample comparison are from \citet{paparoditis2016bootstrap}, who develop an empirical bootstrap methodology and prove its consistency when the test statistics is based on the Hilbert-–Schmidt norm, and \citet{kashlak2019inference}, who perform $K$-sample comparison via concentration inequalities based methods. Recently, \citet{hlavka2022functional} proposed a method to perform functional ANOVA based on empirical characteristic functionals. When comparing with \citet{anderson2006distance, paparoditis2016bootstrap} and \citet{kashlak2019inference}, \citet{cabassi2017permutation} report simulation results illustrating state-of-the-art performance of their method. We found that this holds true even when comparing against the recent work of \citet{hlavka2022functional}. In \citet{hlavka2022functional}, a specific choice of the parameters covariance matrix yield similar conclusion to \citet{cabassi2017permutation} when comparing covariances in a real-data example.

\citet{pigoli2014distances} paid particular attention to the \emph{Procrustes} distance, which generalises a metric used to compare unlabelled shapes into a metric between covariance operators. Heuristically, the Procurstes metric aims to compare roots of two operators in the Hilbert-Schmidt distance, a natural choice since covariances are characterised as squares of Hilbert-Schmidt operators. However, there is an ambiguity as to which precise root one ought to use, and the Procrustes distance corrects for that by optimising over the square root orbits. Indeed, later work by \citet{pigoli2018statistical} reports that the Procrustes metric offered the most natural framework to compare trace-class operators, in that it uses a map from the space of covariance operators to the linear space of Hilbert-Schmidt operators. 
\citet{masarotto2018procrustes} carried out a deeper study of the Procrustes metric and established a fruitful connection with the Wasserstein metric between Gaussian processes. On the one hand, this allowed them to provide a complete geometrical description of the space of covariances under the Procrustes metric, including basic results about Fr\'echet means; on the other hand, it established intriguing parallels with the theory of optimal transportation, offering potentially new avenues and tools for the analysis of covariance operators (see also the discussion of \citet{pigoli2018statistical} by \citet{panaretos2018discussion}).

In this paper, we use precisely this novel transportation perspective to introduce a new ANOVA test, and then exploit the corresponding geometry to construct a PCA that respects the nature of the covariance operators, by representing covariance operators as transport maps. Specifically, we view the testing problem through the lens of optimally multicoupling the Gaussian processes $\{N(0,\Sigma_1),\ldots,N(0,\Sigma_K)\}$, thus translating the task of testing the hypothesis $H_0$ in Equation~\eqref{null_hypothesis} into that of testing whether the optimal multicoupling is ``trivial''.  To do this we first prove that the optimal multicoupling can always be \emph{deterministically}\footnote{Rather than stochastically, by means of a probability measure on $\mathcal{H}^K$.} produced by means of \emph{bounded linear} transport maps $\{\mathbf{t}_j\}$, and this regardless of the validity of the null hypothesis. These two results are of independent interest, and are stated as Theorem~\ref{thm:bounded_maps}. Then, given these results, we translate the task of testing the hypothesis in Equation~\eqref{null_hypothesis} into the equivalent task of testing the hypothesis
\begin{equation}\label{transport_null}
H_0':\mathbf{t}_1-\mathscr{I}=\hdots=\mathbf{t}_K-\mathscr{I}=0
\end{equation}
for $\mathscr{I}$ the identity, and so the hypothesis \eqref{transport_null} can now be tested by means of a norm-based (e.g., operator, Hilbert--Schmidt, or nuclear) test statistic, establishing a direct analogy with classical ANOVA. The $\Delta_j=\mathbf{t}_j-\mathscr{I}$ can heuristically also be viewed as ``roots'' of the original covariances, albeit free of any unitary ambiguity. Though the test is motivated by the 1-to-1 correspondence between covariances and Gaussian processes, it relies in no way on a Gaussian assumption, and can indeed admit an interpretation purely in terms of Procrustean geometry. Our simulation experiments indicate that the new test dominates state-of-the-art competitors, with dramatic gains in power, particularly against more challenging local alternatives. This is also explained by means of theory, and is seen to be a genuinely functional effect, with connection to the Hajek--Feldman condition. See  Section~\ref{sec:conclusions} for more details. 
In terms the computation, the quantity ($K^{-1}\sum_{j=1}^K\mathbf{t}_j - \mathscr{I}$) in finite dimension is precisely the negative gradient of the Fr\'echet (sum-of-squares) functional \citet[Theorem 1]{zemel2019frechet} and its computation is stable and feasible because it relies on the fast nature of the steepest descent algorithm in the space of covariances endowed with the Procrustes metric. 

When the mull hypothesis is rejected, it is natural as a second step to wish to describe the variation manifested by the covariances $\{\Sigma_j\}$, or indeed obtain a parsimonious representation thereof. We show how the maps $\{\mathbf{t}_j\}$ can then readily be employed to do just that, via a principal component analysis. In particular, the $\Delta_j=\mathbf{t}_j-\mathscr{I}$ can be interpreted as the logarithms of the $\{\Sigma_j\}$ at their Procrustes--Fr\'echet mean. The corresponding tangent space admits a Hilbertian structure with respect to a modified Hilbert--Schmidt inner product, which we use to produce a tangent space fPCA, and then retract the principal components back onto the covariance space. To the best of our knowledge, this is the first instance of a functional PCA on covariance operators that respects their intrinsic geometric features as trace-class positive operators. We describe the computational steps required to  do so in Section~\ref{sec:computation}, and illustrate the usefulness of the procedure on simulated data as well as a linguistic data set. The next paragraph collects the notational conventions employed throughout the paper. Proofs of our theoretical results are given in Section~\ref{sec:proofs}.

\section{Methodology}

\subsection{Basic Setting and Notation}

As stated in the introduction, we will be interested in exploring the variation in a finite collection of covariances $\{\Sigma_j\}_{j=1}^{K}$ on a real separable Hilbert space $\mathcal{H}$, equipped with the inner product $\langle\cdot,\cdot\rangle:\mathcal{H}\times\mathcal{H}\rightarrow\mathbb{R}$, and corresponding norm $\|\cdot\|:\mathcal{H}\to[0,\infty)$. Since $\mathcal{H}$ will in principle be infinite-dimensional, we will need to review some basic definitions and notation, which can be more subtle.

Given a bounded linear operator $A:\mathcal{H}\rightarrow\mathcal{H}$, we will denote its trace (when defined) by $\tr  A$ or $\tr( A)$, its adjoint by $A^*$, and its inverse by $A^{-1}$. The inverse may not be defined, or defined only on a (dense) subspace of $\X$. The \emph{range} of $A$ will be denoted by $\mathrm{range}(A)=\{Av:v\in\mathcal{H}\}$ whereas the \emph{kernel} of $A$ will be denoted by $\mathrm{ker}(A)=\{v\in\mathcal H:Av=0\}$.  We will say that a (possibly unbounded) operator $A$ is \emph{self-adjoint} if $\langle Au,v\rangle = \langle u,Av\rangle$ for all $u,v$ in the domain of definition of $A$;  if $A$ happens to also be bounded, then this is equivalent to the condition that $A=A^*$.  

A \emph{non-negative} operator is a self-adjoint, possibly unbounded operator $A$ such that $\langle Au,u\rangle\ge0$ for all $u$ in the domain of $A$.  If in addition $A$ is compact, then there exists a unique non-negative operator whose square equals $A$, which will be denoted by either $A^{1/2}$ or $\sqrt{A}$.  The inverse square root $(A^{1/2})^{-1}$ is denoted by $A^{-1/2}$.  For any bounded operator $A$, $A^*A$ is non-negative.  The identity operator on $\mathcal{H}$ will be denoted by $\mathscr{I}$.  The operator, Hilbert--Schmidt and trace (nuclear) norms will respectively be
\[
\hs A\hs_{\infty}=\sup_{\|h\|=1}\|Ah\|,
\quad \hs A \hs_2 =\sqrt{\tr\left(A^*A\right)},
\quad  \hs A \hs_1=\tr\left(\sqrt{A^*A}\right)
\] 
and can be ordered from coarser to finer as follows
\[
\hs A\hs_{\infty} \le
\hs A\hs_2 \le
\hs A\hs_{1}.
\]
When $\hs A\hs_{2}<\infty$ we say that $A$ is Hilbert--Schmidt and when $\hs A\hs_{1}<\infty$ we say that $A$ is \emph{nuclear} or \emph{trace-class}.

Summarising, in this setting, covariances are linear operators from $\mathcal{H}$ into $\mathcal{H}$, that are \emph{self-adjoint}, \emph{non-negative}, and \emph{trace-class}. As such, a covariance operator $\Sigma$ on $\mathcal{H}$ can be considered as the ``square" of a Hilbert--Schmidt operator: if $\hs B\hs_2<\infty$ then $B$ is certainly bounded, and $B^*B$ defines a valid covariance operator. 

It therefore becomes clear that covariance operators are non-linear objects, and though they can be contrasted by means of any of the three norms $\hs \cdot\hs_{\infty}$, $\hs \cdot\hs_2$, or $\hs \cdot\hs_1$, it may be preferable to find a means of comparison that respects this non-linear nature. In finite dimensions, this is done by means of some form of \emph{linearisation}, i.e. the use of a transformation that substitutes a covariance pair $(\Sigma_1,\Sigma_2)$ to be considered by an operator that can be contrasted to zero by means of one of the norms $\hs\cdot\hs_r$, $r=1,2,\infty$. For instance, in classical two-sampled covariance tests, two covariances $\Sigma_1$ and $\Sigma_2$ have been contrasted by means of quantities such as
$$\hs \Sigma_1\Sigma_2^{-1}-\mathscr{I}\hs_{r}\quad\&\quad   \hs 2\Sigma_2(\Sigma_1+\Sigma_2)^{-1}-\mathscr{I}\hs_{r}$$
assuming that the inverses exist (e.g., \citet{roy1953heuristic},  
\citet{kiefer1965admissible}, \citet{giri1968tests}. In non-Euclidean statistics, covariances $(\Sigma_1,\Sigma_2)$ have been contrasted by means of 
$$\hs \log(\Sigma_1)-\log(\Sigma_2)\hs_r\quad \&\quad \hs \log(\Sigma_2^{-1/2}\Sigma_1\Sigma_2^{-1/2})\hs_r,$$
again, assuming that the inverses exist\footnote{If more covariances $\{\Sigma_i\}$ are to be simultaneously compared, for instance in a covariance ANOVA, one could use the same contrasts, replacing $\Sigma_1$ with $\Sigma_i$ and $\Sigma_2$ by the arithmetic average $n^{-1}\sum_i \Sigma_i$.} (\citet{dryden2009non}).

In infinite dimensions, however, these criteria will generally fail to be well-defined. For example, the inverse of $\Sigma_2$ will be unbounded, and there is no guarantee that $\Sigma_1\Sigma_2^{-1}$ will be bounded, except if $\Sigma_1$ and $\Sigma_2$ share some special relation. Similarly, the logarithm of a covariance operator will typically be unbounded, and unless there is a specific relation between $\Sigma_1$ and $\Sigma_2$, the logarithmic criteria will fail to be well defined.
This is one of the main reasons why much of the literature on covariance operators has focussed on bypassing their nonlinear nature, and comparing them directly, e.g. by means of $\hs \Sigma_1-\Sigma_2\hs_2$ (\citet{panaretos2010second,fremdt2013testing,boente2018testing}).
 
A first step in obtaining linearisations that would yield contrasts respecting the nature of covariances, while being well-defined in infinite dimensions was made by \citet{pigoli2014distances} Since covariances are ``squares" of Hilbert--Schmidt operators, they considered contrasting the square roots of in the Hilbert--Schmidt distance
$$\hs \Sigma^{1/2}_1-\Sigma^{1/2}_2 \hs_2.$$
Observing that one could nevertheless choose roots other than the (unique) positive roots, by means of the fact that $\Sigma_2^{1/2}U(\Sigma_2U)^*=\Sigma_2$ they arrived at the \emph{Procrustes metric}
\[
\Pi(\Sigma_1,\Sigma_2)=\inf_{U^*U=\mathscr{I}}\hs \Sigma^{1/2}_1-\Sigma^{1/2}_2U \hs_2
\]
which lifts the unitary ambiguity by optimising over unitary matrices, and is well defined in both finite and infinite dimensions. Indeed they use this metric to develop a two-sample test for covariance comparison. \citet{masarotto2018procrustes} further developed several key properties of this metric and its geometry, interpreting it via the optimal transportation of Gaussian processes as the $L^2$-Wasserstein distance between two Gaussian measures $N(0,\Sigma_1)$ and $N(0,\Sigma_2)$ on $\mathcal{H}$,
\[
\Pi(\Sigma_1,\Sigma_2)=\inf_{X_i\sim N(0,\Sigma_i)}\mathbb{E}\|X_1-X_2\|^2_2=\hs \Sigma_1\hs_1+\hs \Sigma_2\hs_1-2\mathrm{trace}\left\{\sqrt{\Sigma_1^{1/2}\Sigma_2\Sigma_1^{1/2}}\right\}.
\]
The key observation in this paper is that the optimal transport theory developed in \citet{masarotto2018procrustes} can be directly leveraged in order to provide natural notions of ``roots" (or linearisations) that:
\begin{itemize}
\item are unequivocally defined without any unitary ambiguity;

\item are efficiently computable;

\item that offer remarkable power when used in a covariance ANOVA;

\item can be used in order to obtain a natural PCA, when the equality of covariances is rejected.

\end{itemize}
These are defined via the notion of an optimal multicoupling, and are introduced in the next Section.

\subsection{Optimal Multicoupling and Transport Maps}
As already stated, our strategy for testing $H_0:\Sigma_1=\cdots=\Sigma_K$ is to view the covariance operators through the lens of optimal multicoupling of Gaussian processes. Specifically, we observe that the collection of covariances $\{\Sigma_1,\ldots,\Sigma_K\}$ can be bijectively identified with a collection of centred Gaussian measures $\{N(0,\Sigma_1),\ldots,N(0,\Sigma_K)\}$ on the Hilbert space $\mathcal{H}$. Denote these measures as $\{\gamma_1,\dots,\gamma_K\}$. Equality of the covariance operators thus holds true, if and only if the measures $\{\gamma_j\}$ coincide. Viewing the measures $\{\gamma_j\}$ as the marginals of a joint measure $\pi$ on $\mathcal{H}^K$, one can ask what are the possible forms of $\pi$. This set of possible joint measures $\pi$ is always non-empty (it always contains the product measure), and is called the set of \emph{multicouplings} of $\{\gamma_1,\dots,\gamma_K\}$. An \emph{optimal multicoupling} is a multicoupling $\pi^*$ such that the marginals are as tightly coupled as possible in a pairwise mean-square sense, in that it minimizes the functional
\[
F(\pi)
=\frac{1}{2K^2}\sum_{i,j=1}^{K}
\int_{\mathcal{H}^K} \|x_i-x_j\|^2 \pi(dx_1,\ldots,dx_K).
\]
Said differently, $\pi^*$ is the joint distribution of collection of $K$ Gaussian processes on $\mathcal{H}$, say $(Z_1,\ldots, Z_K)$, such that $Z_j\sim N(0,\Sigma_j)$ marginally for all $j\le K$, while $\sum_{i<j}\mathbb{E} \|Z_i - Z_j\|^2$ is minimized.  Existence of finite second moments of Gaussian measures (a consequence of Fernique's \citet{fernique1970integrabilite} theorem) implies that $F(\pi)$ is finite for any multicoupling $\pi$.  It can be shown that an optimal multicoupling of Gaussians always exists (\citet{masarotto2018procrustes}). We say that such an optimal multicoupling $\pi^*$ is manifested by (deterministic) transport maps if the collection $(Z_1,\ldots, Z_K)\sim \pi^*$ can be generated by taking a single process $Z$, and a collection of deterministic maps $\mathbf t_j:\mathcal{H}\rightarrow \mathcal{H}$ such that
\[
(Z_1,\ldots, Z_K)\stackrel{d}{=}(\mathbf{t}_1(Z),\ldots,\mathbf{t}_K(Z)).
\]
In other words, an optimal multicoupling $\pi^*$ is generated by deterministic maps if  it is supported on the graph of a vector-valued function from $\mathcal H$ to $\mathcal H^K$.  It is a priori unclear whether a deterministic multicoupling exists in general, and if it does, whether the maps $\mathbf{t}_j$ are bounded. but it is not hard to see that it will exist under the null $H_0$ and that it will be ``trivial":
\begin{lemma}\label{lem:nullHyp}
The equality $\Sigma_1=\cdots=\Sigma_K$ holds true if and only if the (unique) optimal multicoupling of $(\gamma_1,\ldots,\gamma_K)$ can be achieved by transport maps satisfying $\mathbf{t}_1=\cdots=\mathbf{t}_K$.
\end{lemma}
The maps $\topt j{}$ are called transport maps because they can be thought of as ``transporting" the (unspecified) law of $Z$ to that of $Z_j$.  The lemma suggests that we can detect departures from the hypothesis $\{H_0:\Sigma_1=\cdots=\Sigma_K\}$ by focussing on departures from the hypothesis $\{H'_0:\mathbf{t}_1=\cdots=\mathbf{t}_K\}$. But to even speak of departures from $H_0'$, we must be assured that such maps exist even under the alternative regime, and this existence is not a priori guaranteed (see Conjecture~17 and the discussion in Section~12 of \cite{masarotto2018procrustes}). Furthermore, to quantify the extent of departures from the null, we need to make sure that the multicoupling maps not only exist, but are bona fide bounded linear operators over all of $\mathcal{H}$, and can thus be contrasted by appropriate norms. 

Our main theoretical result, in the form of the following theorem, shows that a multicoupling can always be realised by means of bounded deterministic maps, a result that is of independent interest in optimal transport in its own right. It provides a partial positive resolution of Conjecture~17 in \cite{masarotto2018procrustes}.
\begin{theorem}\label{thm:bounded_maps}
Let $\{\gamma_1,\ldots,\gamma_K\}$ be an arbitrary finite collection of Gaussian measures on $\mathcal{H}$ with mean zero.  Then there exists an optimal multicoupling of $\{\gamma_j\}_{j=1}^{K}$ manifested by deterministic transport maps $\mathbf{t}_j:\mathcal{H}\rightarrow\mathcal{H}$ that are bounded non-negative linear operators satisfying $\hs \mathbf{t}_j \hs_{\infty}\le K$, for all $j\leq K$.
\end{theorem}

Although the optimal coupling $\pi^*$ is typically unique, its representation in terms of the maps is not.  For instance, if $\pi^*$ is manifested as the law of $(\topt1{}(Z),\dots,\topt{K}{}(Z))$, it may also be represented as $(2\topt1{}(Z/2),\dots,2\topt{K}{}(Z/2))$.  It is natural to take $Z\sim N(0,\overline\Sigma)$, where $\overline\Sigma$ is a centre, i.e., the Fr\'echet mean of $\Sigma_1,\dots,\Sigma_K$ with respect to the Procrustes metric.  This choice forces the maps $\topt{j}{}$ to have mean identity (see \eqref{eq:mean_identity} below), and in particular they must be the identity under the null \eqref{null_hypothesis}, so that \eqref{transport_null} holds.  Using this convention, the existence and boundedness result in Theorem~\ref{thm:bounded_maps} opens the way for a testing procedure: the deviations $\Delta_j=\mathbf{t}_j-\mathscr{I}$ are all self-adjoint and bounded, but no longer restricted to be non-negative. When $H_0$ is valid, Lemma~\ref{lem:nullHyp} implies that $\Delta_j=\mathbf{t}_j-\mathscr{I}=0$ for all $j$. Under the alternative, at least one $\Delta_j$ is non-zero. We can thus replace the null hypothesis
\[
H_0:\Sigma_1=\cdots=\Sigma_K
\]
by the equivalent hypothesis
\[
H_0':\Delta_1=\cdots=\Delta_K=0
\]
viewing the $\Delta_j$ as elements of a linear space, and reducing the original testing problem to a more traditional linear functional ANOVA setting. Since the $\Delta_j$ are guaranteed to be bounded (Theorem~\ref{thm:bounded_maps}), they can certainly be contrasted to 0 using the operator norm. However, one can devise even more powerful procedures by measuring the size of $\Delta_j$ in a stronger norm, such as the  the Hilbert--Schmidt norm, or even the trace norm. In the finite-dimensional setting this choice of norm will typically not make much difference, since all norms are equivalent. But in the infinite-dimensional case, a finer norm will detect subtle departures from the null. For instance, if $K=2$ and $\Sigma_2=\delta^2\Sigma_1$ for some $\delta\ge0$, then one can show that for $j=1,2$, $\hs\Delta_j\hs_{\infty}=|1-\delta|/(1+\delta)\le 1$ while $\hs \Delta_j\hs_{2}=\infty$ unless $\delta=1$.  Similarly, if the $\Delta_j$ are Hilbert--Schmidt but not trace class, their trace norm will be infinite, promising to furnish high power even against very local alternatives. This genuinely functional phenomenon is not unlike the possibility of perfect discrimination of Gaussian processes (\citet{feldman1958equivalence,hajek1958property,rao1963discrimination}; see Section \ref{sec:conclusions} for a more detailed discussion). It is demonstrated empirically in our later simulations.  If there are only $K=2$ populations, then in view of \eqref{eq:mean_identity}, $\Delta_2=-\Delta_1$ and the test statistic is $\hs\Delta_1\hs_r$.  When comparing more than two covariance operators, the criteria $\hs \Delta_j \hs_r$ will need to be combined into a single criterion (e.g., by taking their supremum over $j$ or by summation).

How does one concretely construct a deterministic multicoupling $\{\mathbf{t}_j\}$, and hence the $\{\Delta_j\}$ in practice?  In proving Theorem~\ref{thm:bounded_maps} we establish the existence and boundedness of the maps
\begin{equation}\label{eq:multi_def}
\mathbf{t}_j
={\Sigma}^{-1/2}( {\Sigma}^{1/2}\Sigma_j {\Sigma}^{1/2})^{1/2} {\Sigma}^{-1/2},
\quad j=1,\ldots,K,
\end{equation}
where ${\Sigma}$ is a Fr\'echet mean of $\{\Sigma_j\}_{j=1}^{K}$ with respect to the procrustes metric $\Pi$, i.e., a minimiser of the sum-of-squares functional $\Gamma\mapsto \sum_{j=1}^{K}\Pi^2(\Gamma,\Sigma_j)$ over the space of trace-class covariances.  Moreover, the $\topt j{}$ are centred around the identity in that
\begin{equation}\label{eq:mean_identity}
\frac{1}{K}\sum_{j=1}^{K}\mathbf{t}_j
=\mathscr{I}.
\end{equation}
The Fr\'echet mean $\Sigma$ is unique when at least one of the $\Sigma_j$ is injective (or more generally, if the kernel of at least one $\Sigma_j$ is contained in the kernels of all other $\Sigma_j$); see \citet[Proposition~10]{masarotto2018procrustes}. Its algorithmic construction is discussed in detail in Section~\ref{sec:computation}. 

Once the multicoupling \eqref{eq:multi_def} has been constructed, and a norm $\hs\cdot\hs_r$ ($r=1,2,\infty$) has been chosen, the null hypothesis can be tested by measuring a combined deviation of the $\Delta_j$ from zero using that norm, and calibrating the typical values of such deviations under the null. This is discussed in the next subsection.  As discussed in \citet{masarotto2018procrustes}, the optimal multicouping has an elegant geometrical interpretation in terms of the manifold geometry of the Procrustes distance --- this will later be exploited in Section~\ref{pca_section} in order to construct a functional PCA of the covariance operators.

\subsection{Transportation-Based Functional ANOVA of Covariances}
Assume now that we have $K$ independent groups of functional data $\{X_{ij},\ j=1,\dots,n, \ i=1, \dots,K\}$, each having covariance $\Sigma_i$ (the procedure can be easily adapted to different group sizes).  Without loss of generality, the data are assumed to be zero mean. Based on the discussion in the previous paragraph, we can test the equality of covariance operators by means of testing the hypothesis
\[
H_0':\underset{=\Delta_1}{\underbrace{\mathbf{t}_1-\mathscr{I}}}=\cdots=\underset{=\Delta_K}{\underbrace{\mathbf{t}_K-\mathscr{I}}}=0
\]
where
$$
\mathbf{t}_j={\Sigma}^{-1/2}({\Sigma}^{1/2}\Sigma_j{\Sigma}^{1/2})^{1/2}{\Sigma}^{-1/2},\quad j=1,\ldots,K,
$$
and ${\Sigma}$ is a Fr\'echet mean of $\{\Sigma_j\}_{j=1}^{K}$ with respect to the procrustes metric $\Pi$. At the level of our sample, we have access to empirical versions of $\{\hat\Sigma_j\}$, constructed on the basis of the samples of size $n$ from each group. These could simply be the empirical covariances within each group (under a complete observation assumption), or some smoothed estimator (for instance the empirical covariance of smoothed versions of the $\{X_{ij}\}$ (\citet{ramsay-silverman-2005}), or PACE-type estimators (\citet{yao-etal-2005a})). Whichever the case may be, the $\hat\Sigma_j$ are finite dimensional, of rank $q\leq n$. In case a smoothing technique is used, we assume that it is such that the $\hat\Sigma_j$ share a common range, and can thus be represented as $q\times q$ positive matrices, via a common (tensor product) basis.  For tidiness, we use the same notation for $\hat\Sigma_j$ and its $q\times q$ matrix representation in the common basis. It is clear that this basis can be chosen so that at least one of these matrices is of full rank $q$. 

In this case, there exists a unique empirical Fr\'echet mean $\hat\Sigma$,
\[
\hat\Sigma=\underset{\mathbb{R}^{q\times q}\ni\Gamma\succeq 0}{\arg\min}\sum_{j=1}^{K}\Pi^2(\hat\Sigma_j,\Gamma)
\]
and this can be computed from $\{\hat\Sigma_1,\ldots,\hat\Sigma_K\}$ using steepest descent (see Section~\ref{sec:computation}).  This gives rise to empirical versions of the $\mathbf{t}_j$,
\[
\hat{\mathbf{t}}_j=\hat{\Sigma}^{-1/2}(\hat{\Sigma}^{1/2}\hat\Sigma_j{\hat\Sigma}^{1/2})^{1/2}{\hat\Sigma}^{-1/2},\quad j=1,\ldots,K,
\]
and corresponding empirical deviations from the identity
\[
\hat\Delta_j=\hat{\mathbf{t}}_j-I_{q\times q}.
\]
The testing procedure is now based on the test statistic
\[
T_r=\sum_{j=1}^{K}\hs \Delta_j\hs_r^2,
\]
where $r\in\{1,2,\infty\}$ (the performance under the different choices of $r$ is investigated in the simulation section).  We avoid making any concrete parametric assumptions, and instead calibrate the test statistic by means of permutations.  The typical permuted value will be calculated according to the following steps:
\begin{itemize}
\item[-] Reassign the $n\times K$ curves $\{X_{i,j}\}$ into $K$ groups of equal size. Call these new groups $\{X^*_{i,j}\}$.
\item[-] Construct the empirical covariance ${\hat\Sigma^*_{j}}$ for the $j$th group $\{X^*_{i,j}\}_{i=1}^{n}$, $j=1,\ldots,K$.
\item[-] Compute empirical Fr\'{e}chet mean $\hat{\Sigma}^*$ of $\{\hat\Sigma^*_{1},\ldots,\hat\Sigma^*_{K}\}$.
\item[-] Construct $\hat{\mathbf{t}}_j^*=(\hat{\Sigma}^*)^{-1/2}\sqrt{(\hat{\Sigma}^*)^{1/2}\hat\Sigma_j^*({\hat\Sigma}^*)^{1/2}}({\hat\Sigma}^*)^{-1/2}$ and compute 
\[
T_r^*
=\sum_{j=1}^{K}\hs \hat{\mathbf{t}}_j^*-I_{q\times q}\hs_r^2=\sum_{j=1}^{K}\hs \hat{\Delta}_j^*\hs_r^2.
\]
\end{itemize}
Repeating these steps for all possible re-assignments yields the distribution for the permuted statistics $T^*_r$, which can be used to generate a $p$-value for $T_r$ under the null hypothesis. As usual, an exact such $p$-value can become prohibitive for large $K$, and in practice we resort to Monte Carlo sampling of permutations. Note that similar steps allow for the implementation of a bootstrap-type procedure, simply by randomly permuting indices with replacement. However, we opt for the permutation approach since the exchangeability of the permutation labels under $H_0$ guarantees the (near) exactness of the $K$-sample permutation test (\citet{pesarin2010permutation}), and we do not pursue the bootstrap approach further.

\subsection{Transportation-Based Functional PCA of Covariances}\label{pca_section}\label{sec:MethPCA}
When the null hypothesis of equality between covariances is rejected, the analyst may wish to explore whether the detected differences are carried by some interpretable main modes of variation. The transport maps $\mathbf{t}_j$ (or their empirical versions) can be used to this aim. When these exist, the differences $\Delta_j=\mathbf{t}_j-\mathscr{I}$ admit an elegant geometric interpretation as the logarithms of the operators $\Sigma_j$ at the Fr\'echet mean $\Sigma$, under the manifold-like geometry induced by the Procrustes metric $\Pi(\cdot,\cdot)$ on the space of (trace-class) covariance operators. Specifically, \citet{masarotto2018procrustes} show that it admits a tangent space with respect to geometry induced by $\Pi$ that is characterised as
\[
\Tan_\Sigma=\overline{\left\{
Q:  Q=Q^*,\, \hs \Sigma^{1/2} Q\hs_2<\infty\right\}}
\]
where the closure is with respect to the inner product
\[
\langle Q_1,Q_2\rangle_{\Sigma}=\mathrm{trace}(Q_1\Sigma Q_2).
\]
When $\Sigma$ is injective, this is a bona fide inner product, that is, $\innprod QQ=0\iff Q=0$.  Unfortunately, whether the Fr\'echet mean is injective or not remains an open question, even if all $\Sigma_j$ are so (see \cite[Conjecture~17]{masarotto2018procrustes}), and it is not clear that the $\Sigma_j$'s can be lifted to the tangent space.  Nevertheless, our new result in the form of Theorem~\ref{thm:bounded_maps} guarantees that the maps $\Delta_j$ do exist as bounded self-adjoint operators, and indeed the 1-form
\[
\mathrm{trace}(\Delta_i\Sigma \Delta_j)
\le  \hs\Sigma^{1/2}\Delta_i \hs_2 \hs\Sigma^{1/2}\Delta_j \hs_2
= \Pi(\Sigma_i,\Sigma)\Pi(\Sigma_j,\Sigma)
<\infty
\]
is well-defined, regardless of the injectivity of $\Sigma$ by means of the formulae for $\mathbf{t}_i$ and $\Pi$. Consequently, the finite-dimensional span of $\{\Delta_1,\ldots,\Delta_K\}$ admits a Hilbertian structure when equipped with the inner product $\langle \cdot,\cdot\rangle_{\Sigma}$, and for all practical purposes can be used to carry out a PCA\footnote{Such a PCA can be interpreted as a tangent space PCA with respect to a \emph{Procrustean metric tensor} 
\[
\langle Q_1,Q_2\rangle_{\Gamma}
=\mathrm{trace}(Q_1\Gamma Q_2)
,\quad \Gamma\in \mathscr{L}
=\left\{\underset{\Gamma\succeq 0}{\arg\min}\sum_{j=1}^{K}\alpha_j\Pi^2(\hat\Sigma_j,\Gamma):\alpha_j>0 \,\&\,\sum_{j=1}^K\alpha_j=1\right\}
\]
over the barycentric locus $\mathscr{L}$ of the operators $\{\Sigma_j\}_{j=1}^{K}$.} based on the spectral decomposition of the non-negative operator
\[
\mathscr{K}
=\frac{1}{K}\sum_{j=1}^K\Delta_j\otimes_{\Sigma}\Delta_j
=\frac{1}{K}\sum_{j=1}^K
\left(\topt j{} - \mathscr{I}\right)\otimes_{\Sigma} \left( \topt{j}{} - \mathscr{I}\right),
\]
where $(A\otimes_{\Sigma}B)C=\langle B,C\rangle_{\Sigma}A$. Notice that the latter constitutes precisely the empirical covariance of the collection $\{\Delta_j\}_{j=1}^{K}$, because
\[
\sum_{j=1}^K\Delta_j=0,
\]
by Equation~\eqref{eq:mean_identity}. Once the principal components are constructed, the main modes of covariance variation can be visualised by retracting appropriate subspaces of the tangent space back to the space of covariance operators.  Specifically, if $E_1$ is the eigenoperator associated with the largest eigenvalue of $\mathscr{K}$, this retraction takes the form
\[
t\mapsto (\mathscr{I}+tE_1)\Sigma(\mathscr{I}+tE_1),\qquad t\in[-\epsilon,\epsilon],
\]
which is a geodesic for sufficiently small $\epsilon>0$.  This principal geodesic is the visualisation of the main mode of variation of $\{\Sigma_j\}_{j=1}^K$ near their Fr\'echet mean $\Sigma$.

A subtlety here is that the PCA is to be carried out on a Hilbert space endowed with an inner product other than the standard Hilbert--Schmidt inner product. This different choice of inner product affects both the formal definition and the computational evaluation of the principal components. The defining maximisation problem yielding the first principal component is now
\[
\underset{{\hs B\hs_{\Sigma}=1}}{\arg\max}\langle \mathscr{K}B,B\rangle_{\Sigma}
=\underset{{\hs \Sigma^{1/2}A\hs_{2}=1}}{\arg\max}\langle \mathscr{K}\Sigma^{1/2}A,\Sigma^{1/2}A\rangle_{2}
=\underset{{\mathrm{trace}(A\Sigma A)=1}}{\arg\max}\mathrm{trace}(\mathscr{K}\Sigma^{1/2}A^2\Sigma^{1/2}).
\]
Nevertheless, this change of inner product poses no essential difficulty, and has indeed considered before by \citet{silverman1996smoothed} in the case of Sobolev inner products, and generalised by \citet{ocana1999functional}.  Further details are given in Section~\ref{sec:computation}.

\section{Computational Implementation}\label{sec:computation}
In the next Section we will work with $K$ independent groups of functional data $\{X_{ij},\
j=1,\dots,K, \ i=1, \dots,n_j\}$, each group of sample size $n_j$ and with covariance operator $\Sigma_j$, $j=1, \dots,K$. Unless otherwise stated, all curves are curves are simulated from a multivariate Gaussian process and sampled on an equispaced grid on $\Omega=[0,1]$. The sample size and the grid points vary across applications, therefore, in practice, we only have access to estimated empirical covariances $\widehat{\Sigma}_1,\dots,\widehat{\Sigma}_K$. In our case, $\widehat{\Sigma}_1,\dots,\widehat{\Sigma}_K$ are obtained from the smoothed versions of the $\{X_{ij}\}$ as traditional sample covariance functions, through the command \texttt{var.fd} in the \texttt{R} package \texttt{fda} (\citet{ramsay2005springer,ramsay2021package}). If a smoothed version of the $\{X_{ij}\}$'s is not available, a PACE-type estimator can be used (\citet{yao-etal-2005a}).

\subsection{Transport ANOVA}
Once estimators $\widehat{\Sigma}_1,\dots,\widehat{\Sigma}_K$ are at our disposal, our transport ANOVA requires their Fr\'echet mean
\[
\overline\Sigma
=\arg\min_\Sigma \sum_{j=1}^K \Pi^2(\Sigma,\widehat{\Sigma}_j)
\]
and the transport maps contrasted with the identity
\[
\Delta_j
=\topt j{} - \mathscr I
={\overline\Sigma}^{-1/2}( {\overline\Sigma}^{1/2}\widehat{\Sigma}_j{\overline\Sigma}^{1/2})^{1/2} {\overline\Sigma}^{-1/2} - \mathscr I,
\quad j=1,\ldots,K.
\]
When $\widehat{\Sigma}_j$ commute ($\widehat{\Sigma}_j\widehat{\Sigma_i}=\widehat{\Sigma_i}\widehat{\Sigma}_j$ for all $i,j$), the $\overline\Sigma$ has the explicit form
\[
\overline\Sigma^{1/2}
=K^{-1}\left[\widehat{\Sigma}_1^{1/2} + \dots + \widehat{\Sigma}_1^{1/2}\right].
\]
However, there is no reason that these commutativity should hold.  For general covariances, $\overline\Sigma$ has no closed form formula, but it can be approximated by the iterative procedure described in \cite[Section~8]{masarotto2018procrustes}.  It can be interpreted as steepest descent in the Procrustes space of covariances, and in finite dimensions provably approximates $\overline\Sigma$ and $\Delta_j$ to arbitrary precision.  It is carried out as follows:
\begin{itemize}
\item Let $\Sigma^0:\mathcal{H}\rightarrow\mathcal{H}$ be an injective covariance, serving as the initial point. 

\item Denote the current iterate at step $k$ as $\Sigma^k$.  

\item For each $j$ compute the optimal maps from $\Sigma^k$ to each of the prescribed operators $\widehat{\Sigma}_j:\mathcal{H}\rightarrow\mathcal{H}$, namely 
\[
\topt{\Sigma^k}{\widehat{\Sigma}_j}
=(\Sigma^k)^{-1/2}[(\Sigma^k)^{1/2}\widehat{\Sigma}_j(\Sigma^k)^{1/2}]^{1/2}(\Sigma^k)^{-1/2}.
\]
\item Define their average $T_k=K^{-1}\sum_{j=1}^K \topt{\Sigma^k}{\widehat{\Sigma}_j}$, which is itself a non-negative operator on $\mathcal{H}$.

\item Set the next iterate to $\Sigma^{k+1}=T_k\Sigma^kT_k$.
\end{itemize}
In practice the algorithm will stop after, say, $k$ iterations, $\Sigma^{k}$ will be our numerical approximation for $\overline\Sigma$ and $\topt{\Sigma^k}{\widehat{\Sigma}_j}-\mathscr I$ will approximate $\topt j{}-\mathscr I$.

In terms of the manifold-like geometry of covariances under the Procrustes metric (see Section~\ref{sec:MethPCA}), the algorithm starts with an initial guess of the Fr\'echet mean; it then lifts all observations to the tangent space at that initial guess via the log map, and averages linearly on the tangent space; this linear average is then retracted onto the manifold via the exponential map, providing the next guess, and iterates.  The quantity $T_k-\mathscr{I}$ is precisely the \emph{negative gradient} of the Fr\'echet (sum-of-squares) functional, which is the reason why this is steepest descent. 
The test statistic, a linear combination (or the maximum) of powers of $\hs \Delta_j \hs_2^2$, is readily computable from the spectral decomposition. 

\subsection{Transport PCA}\label{subsec:transport_pca}
Once the empirical Fr\'echet mean $\overline\Sigma$ of $\widehat{\Sigma}_1,\dots,\widehat{\Sigma}_K$ and the $\widehat{\Delta}_j$'s are computed (see the previous section), we can perform functional PCA on the collection $\{\widehat{\Sigma}_1,\dots,\widehat{\Sigma}_K\}$ by their tangent space representation $\widehat{\Delta}_1,\dots,\widehat{\Delta}_K$ (see Section~\ref{sec:MethPCA}).  As explained there, there are some subtleties involved in this PCA, since we are working with a different inner product than the standard Hilbert--Schmidt one.  However, \citet{ocana1999functional} proved that PCA with respect to the tangent space inner product is equivalent to the PCA performed with the Hilbert--Schmidt inner product on suitably transformed data, thus allowing a framework to interpret standard Euclidean PCA in the Procrustes geometry.  More precisely, let $\langle \cdot ,\cdot \rangle_{HS}$ be the Hilbert--Schmidt inner product and $\langle \cdot,\cdot \rangle_{\Sigma}$ be the Wasserstein one at the tangent space at $\bar\Sigma$, that is $\langle A,B\rangle_\Sigma=\mathrm{tr}(A{\Sigma}B)$.  We follow the steps by \citet{ocana1999functional}, using that there is a unique operator $\mathscr T$ characterised by 
\[
\langle A,B\rangle_\Sigma
=\langle \mathscr T(A),B\rangle_{HS}
=\mathrm{tr}([\mathscr T(A)]^*B).
\]
which in our case we take to be the multiplication from the right by $\Sigma$ (so $\mathscr T(A)=(A\Sigma^{1/2})\Sigma^{1/2}$ is trace class and has an adjoint).  
$\mathscr T$ is nonnegative and the PCA of some data $(\mathscr X_1,\dots,\mathscr X_n)$ with respect to $\langle \cdot,\cdot\rangle_\Sigma$ 
is computed as the PCA of $[\mathscr T^{1/2}(\mathscr X_i)]_{i=1}^n$ with Hilbert--Schmidt norm, in the sense that the eigenvalues (i.e. the variances) remain the same and the eigenfunctions with respect to $\innprod\cdot\cdot_\Sigma$ are $\mathscr T^{1/2}$ applied to the eigenfunctions with respect to $\innprod\cdot\cdot_{HS}$.  In our specific case, $\mathscr T^{1/2}(\mathscr X)=\mathscr X\Sigma^{1/2}$, and the PCA on the tangent space is carried out as follows:
\begin{itemize}
\item[-] Multiply $\Delta_j=\topt j{} - \mathscr{I}$ from the right by $\Sigma^{1/2}$.
\item[-] Find the spectral decomposition of the empirical operator $$\tilde{\mathscr K}=K^{-1}\sum \Delta_j\Sigma^{1/2}\otimes \Delta_j\Sigma^{1/2},$$ defined on the space of Hilbert--Schmidt operators with respect to the Hilbert--Schmidt norm.
\item[-] Multiply (from the right) the eigenfunctions of $\tilde{\mathscr K}$ by $\Sigma^{-1/2}$ to obtain the eigenfunctions of $\mathscr K$.
\end{itemize}

\section{Numerical Experiments}\label{sec:Applications}

We now demonstrate the efficacy of the proposed methods through a variety of simulated examples. Simulations are broadly categorized into two subsets: functional ANOVA and tangent space PCA.
We initially explore the behavior in a scenario where the Fr\'echet mean is known and the transport-based functional ANOVA test is applied to covariances $\{\Sigma_j=T_j\Sigma T_j\}_{j=1}^K$ obtained via perturbation according to the generative model described in \citet[Section~10]{masarotto2018procrustes}.

To allow comparability of our method with the existing literature, we subsequently consider another simulation scenario, directly taken from \citet{cabassi2017permutation}, where the simulated covariance operators are perturbation of the male and female subjects in the Berkeley growth data set (\citet{jones1941berkeley}). In both scenarios, generative model and Berkeley data, we compare the functional ANOVA based on Transport Maps with the permutation test of \citet{cabassi2017permutation}, the concentration inequality method by \citet{kashlak2019inference} and the functional ANOVA method based on empirical characteristic functionals of \citet{hlavka2022functional}.  \citet{cabassi2017permutation} provide result showing that their method is state-of-the-art, and to the best of our knowledge, no other alternative procedures besides the ones listed exist. 
It is evident from Figure \ref{fig:ksample-generative} and \ref{fig:power-Gaussian} that our testing method over-powers other methods, reaching nearly perfect power even in the presence of small differences. 

After validating the performance of transport-based functional ANOVA,  in Section \ref{subsec-tgPCA} we make use of the generative model scenario to perform tangent space Principal Component.
Finally, in \ref{sec-dataanalysis} we test the performance of our method on the classic phoneme dataset (\citet{ferraty:2006}), which consists of 4509 log-periodograms of 5 different phonemes. The data is available at \texttt{https://hastie.su.domains/ElemStatLearn/}. Real data analysis consolidates the strength of our method with respect to the competitors, as it can be seen in Section \ref{subsec:anova-phoneme}. If the null hyphothesis of equality among covariance operators is rejected, Section \ref{subsec:pca-phoneme} shows how tangent space PCA can be a successful tool in understanding dataset variability. 

\subsection{Simulation Experiments}\label{subsec-simulation-exp}
In order to avoid propagation of error, it is convenient to formulate a simulation setup in which the Fr\'echet mean is known exactly and does not need to be approximated (Section~\ref{sec:computation}).  It is easy to construct such examples in the commutative case, but we shall not do so, as this case is overly restrictive and unrealistic.  We thus appeal to the generative model in \citet[Section~10]{masarotto2018procrustes}, which states that if a collection of nonnegative maps $T_1,\dots,T_K$ has mean identity, then any covariance operator $\Sigma$ is the Fr\'echet mean of $\{\Sigma_j=T_j\Sigma T_j\}_{j=1}^K$, and the maps $\topt j{}$ in \eqref{eq:multi_def} must equal $T_j$ (on the closed range of $\Sigma$).

To construct the collection $\{T_1,\dots,T_K\}$ we proceed as follows.  Let $\X=L^2[0,1]$ and for $f,g\in\X$ their tensor product $f\otimes g$ is the operator
\[
(f\otimes g)(h)
=\innprod fhg
=\left( \int_0^1f(t)h(t)dt \right) g \in \X,
\qquad h\in \X.
\]
If $t$ is the parameter of the functions, we write $f(t)\otimes g(t)$ to mean $f\otimes g$.  With this notation set
\begin{equation}\label{eq:generativeModel}
T_j = k^{-1}\sum_{n=1}^\infty \delta_n^{(j)} \sin(2n\pi t - \theta^{(j)}) \otimes \sin(2n\pi t - \theta^{(j)})
, \quad j\in\{1,\dots,K\}
,\quad \delta_n^{(j)}\stackrel{iid}{\sim} \chi^2_k,
\end{equation}
where $\delta_n^{(j)}$ are independent of $\theta^{(j)}$, and $k>0$.  This construction guarantees that $\mathbb E [T_j]=\mathbb E[\mathbb E[T_j|\theta^{(j)}]=\mathscr I$ regardless of the distribution of $\theta^{(j)}$.  The parameter $k$ controls the concentration of $T^{(j)}$ around the identity;  when $k$ is large, the law of large numbers entails that $\delta_n^{(j)}$ is close to 1.  Of course, a given realisation of $T_1,\dots,T_K$ will not average precisely to the identity, but will average approximately to the identity if $K$ and $k$ are not too small.  The parameter $\theta_i\geq 0$ on the other hand, serves as indicators on how far
  we are
  from commutativity. On this note, a parametric model can be assumed for the $\theta_i$. We
chose the $\theta_i$ to be
sampled from a von Mises distribution with mean $0$ and measure of concentration
$1/\sigma$, with the degenerate case of $\sigma\to\infty$ yielding commutativity. 

We then generate $n_j$ Gaussian curves $X_{i,j}$, $i\in\{1, \dots, n_i\}$ with mean zero and covariance $\Sigma_j=T_j\bar\Sigma T_j^\star$, $j\in\{1,\dots,K\}$, $j\in\{1,\dots,K\}$.  Inspired from \citet{kashlak2019inference}, the ``population" Fr\'echet mean was chosen to be a matrix with eigenvalue decay rate $O(n^{-4})$:
\begin{equation}\label{eq:origin}
\bar\Sigma
=U\left[
\sum_{n=1}^\infty n^{-4} \sin(2n\pi t) \otimes \sin(2n\pi t)\right]
U^*
\end{equation}
where $U$ is a randomly generated orthogonal operator. To simulate a functional case and have enough information to display the decay of the spectrum, we chose for the matrices a relative high size of 50$\times$50. The power is estimated from 500 replications. The
number of permutations is 200. At each replication, we generate two optimal maps $T_1$
and $T_2$ via the generative model, and two corresponding covariances
$\Sigma_1=T_1\Sigma T_1^*$ and $\Sigma_2 =T_2\Sigma T_2^*$,
  with $\Sigma$ given by equation~\eqref{eq:origin}. For each
  $\Sigma_i$, $i=1,2$, we sample
  40 observations of a Gaussian process with mean-zero and covariance
  $\Sigma_i$. The empirical covariance computed from these
  observations will yield a replica of  $\Sigma_i$.
We repeat this as to obtain $k_1$ replicas of $\Sigma_1$, and $k_2$ replicas of
    $\Sigma_2$, for a total of $k_1+k_2$ covariances divided into two groups of size
$k_1$ and $k_2$ respectively. The values of the pair $(k_1,k_2)$ are
(1,2), (1,3), (1,7) and (4,4). This procedure is repeated for several
values of the Von Mises parameter, namely $\sigma^{-1}=(0.1,1,5,10)$. \\

We compare the power of our procedure with that of the $K$-sample permutation test in \citet{cabassi2017permutation, kashlak2019inference,hlavka2022functional}.   The idea in \citet{cabassi2017permutation} is to perform a series of partial $2$-sample tests for each pair of groups, and combine the pairwise test statistics through the non-parametric combination algorithm of \citet{pesarin2010permutation}.  The pairwise test statistics are $T_{ij}=d(\Sigma_i,\Sigma_j)$, where $\Sigma_i$ and $\Sigma_j$ are the sample covariance operators of the corresponding groups, and $d$ is a metric on the operators space.  The method of \citet{cabassi2017permutation} is general, as any distance $d$ can be used as test statistic.  It is shown to be more powerful than competing method such as \citet{kashlak2019inference}, and is implemented in the R-package \texttt{fdcov} (\citet{cabassi2016fdcov}).  In our comparisons, we have used the square root distance $\hs \Sigma_i^{1/2} - \Sigma_j^{1/2}\hs_2$, which led to the best performance according to \citet{cabassi2017permutation}.  Using the Procrustes distance instead of the square root distance had similar performance but increased computational cost. \citet{hlavka2022functional} propose a a test statistics of Cram\'er-von Mises type with the distance of the empirical characteristics functionals (ECFs) $\int|\phi_1(\omega)-\phi_2(\omega)|^2\textrm{d}Q(\omega)$, where $\phi_i$ is the ECF of the sample and $Q$ is some probability measure on the dual space of $\mathcal{H}$. The performance of the proposed test by \citet{hlavka2022functional} relies heavily on the choice of the measure $Q$. We consider a Gaussian measure characterized by the choice of two different covariance operators: the sample covariance matrix on one hand and an approximation of the inverse of the pooled sample covariance matrix computed from the first 9 eigenvectors on the other. These two cases appear to be performing the best in \citet{hlavka2022functional}. We include both choices in our simulation study because the first yields overall better performance, but the latter gains power in difficult cases, like when most operators are equals. 
After performing the global test, in case the null hyphothesis $H_0$ is rejected, one can investigate pairwise differences with post-analysis comparison, as in \citet{cabassi2017permutation,pesarin2010permutation}.

Figure \ref{fig:ksample-generative} shows the empirical power of all procedures.  The $x$-axis represents the value of the dispersion parameter $k$ in \eqref{eq:generativeModel}, while the $y$-axes displays the empirical power.  It is evident from the figure that our procedure over-powers that of \citet{cabassi2017permutation, kashlak2019inference} as well as of \citet{hlavka2022functional} for both variations considered. 
\begin{center}
\begin{figure}
\includegraphics[width=8cm]{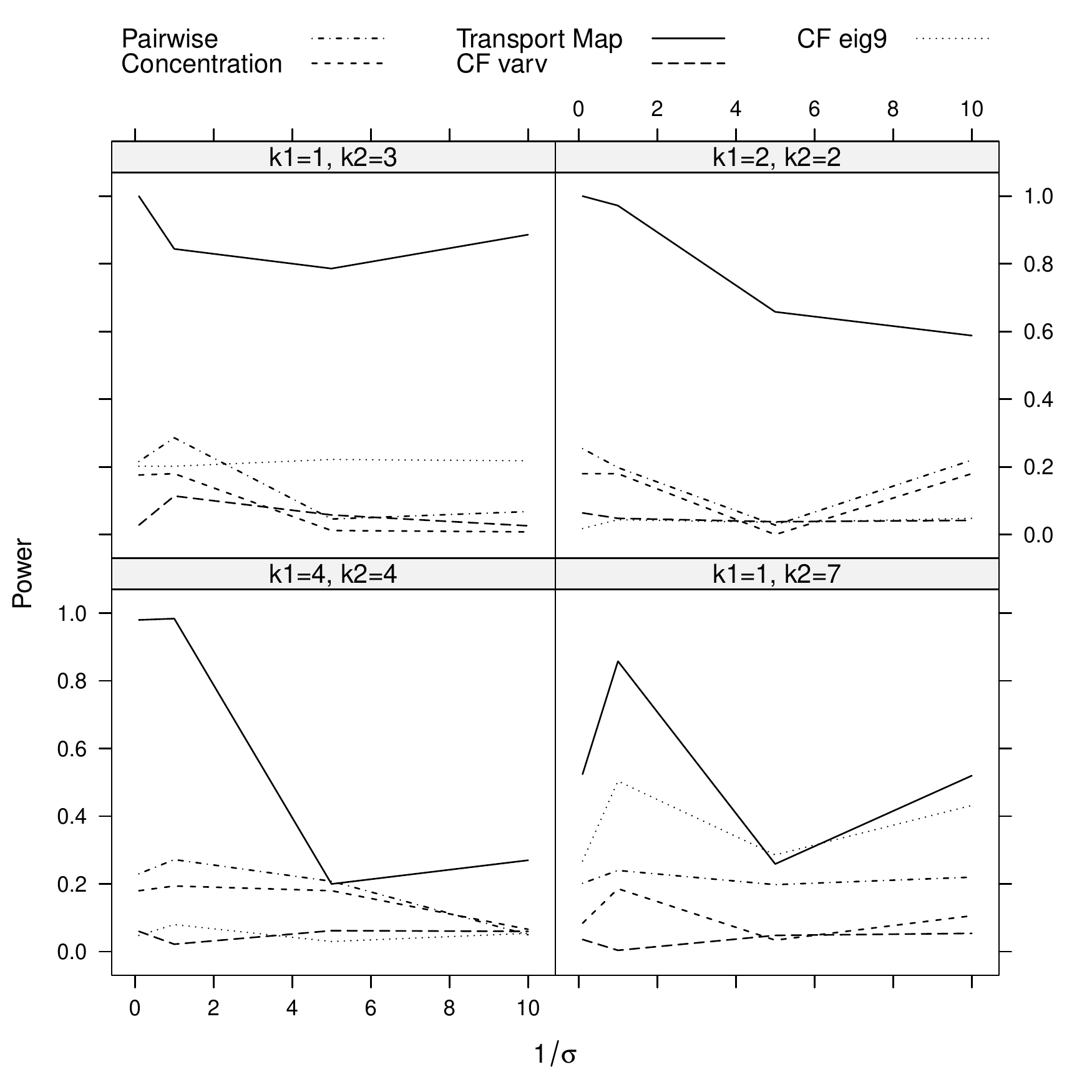}
\caption{Empirical power of different tests as a function of the dispersion parameter $k$ in \eqref{eq:generativeModel}.  The dotted horizontal line represents the nominal level $\alpha=0.05$. }
\label{fig:ksample-generative}
\end{figure}
\end{center}
The success of our test is a genuinely \emph{functional} phenomenon and indeed, when the data are truncated, the differences are not so overwhelming.  To understand this better notice that, since $\chi^2_k$ is an unbounded random variable we have $\hs T_j\hs_\infty=\infty$ almost surely, and our test based on $\hs T_j - \mathscr I\hs_2=\infty$ consequently rejects the null hypothesis.  If the series is truncated at a finite level $n_0$, then $\hs T_j\hs_1\sim k^{-1}\chi^2_{k,n_0}$ and so $P(\hs T_j\hs_1>R)$ decays to zero exponentially as $R$ and/or $k$ increase;  a fortriori the same holds for $\hs T_j\hs_2$ and $\hs T_j\hs_\infty$.  The procedure of \citet{cabassi2017permutation} effectively puts very small weights on what happens at the tails ($n$ large), whereas our procedure is able to detect departures from the null even when they only occur at very high frequencies. We refer to Section \ref{sec:conclusions} for more insight on the power of the functional ANOVA test. 

One may argue that the generative model setup in \eqref{eq:generativeModel} artificially favours our testing procedure, as it guarantees $\hs \Delta_j\hs=\infty$.  We therefore consider another simulation scenario, directly taken from \citet{cabassi2017permutation}, in order to compare the two methods on the same ground.  Here $n=20$ curves generated from a mean-zero Gaussian process with suitably chosen covariances are evaluated on a equipaced grid of 31 points on $[0,1]$. Such curves are assumed to come from $K$
populations, and the covariance operators of each population is obtained via perturbations of some given, known covariances $\Sigma_f$ and $\Sigma_m$.  
These are computed with the \texttt{fda} \textsf{R}-package as the covariance operators of the smoothed growth curves of male and female subjects in the Berkeley growth data set (\citet{jones1941berkeley}).  \citep{jones1941berkeley} The perturbations take two different forms:
\begin{enumerate}
\item \emph{geodesic perturbations:} $K_1<K$ of the groups have covariance operator
\[
\Sigma(\gamma)=[\Sigma_m^{1/2}+\gamma(\Sigma_f^{1/2}R-\Sigma_m^{1/2})][\Sigma_m^{1/2}+\gamma(\Sigma_f^{1/2}R-\Sigma_m^{1/2})]^*
\]
with $R$ the operator minimising the procrustes distance and $\gamma\in[0,5]$.  The other $K_2=K-K_1$ groups have covariance operator $\Sigma_m$.

\item \emph{additive perturbations:} $K_1<K$ of the groups have covariance operator $\Sigma(\gamma)=(1+\gamma)\Sigma_m$, $\gamma\in[0,5]$.  The other $K_2=K-K_1$ groups have covariance operator  $\Sigma_m$. 
\end{enumerate}

The number of permutation is again 200. The power is estimated from a total of 500 replications. The test-statistics employ the Hilbert-Schmidt norm ($r=2$). The probabilities of false positive (I type error) are estimated using all the available replications when $\gamma=0$.
Figure~\ref{fig:power-Gaussian} compares the power of our transport test with that of \citet{cabassi2017permutation, kashlak2019inference, hlavka2022functional} on these synthetic data. The $x$-axis gives the value of the $\gamma$ parameter, while the $y$-axes displays the empirical power. It is seen that our method is more powerful in all scenarios considered. Moreover, in case of geodesic perturbations, we achieve near perfect power, as opposed to the other tests that have little to nearly no power, for small values of $\gamma$, i.e. against local alternatives. Furthermore, notice that without knowing the null distribution is not possible to use the calibration procedure of \citet{kashlak2019inference}, which is too conservative and does not respect the nominal level of $0.05$ under $H_0$.

\begin{center}
\begin{figure}
  \begin{minipage}{0.49\linewidth}
\includegraphics[width=\textwidth]{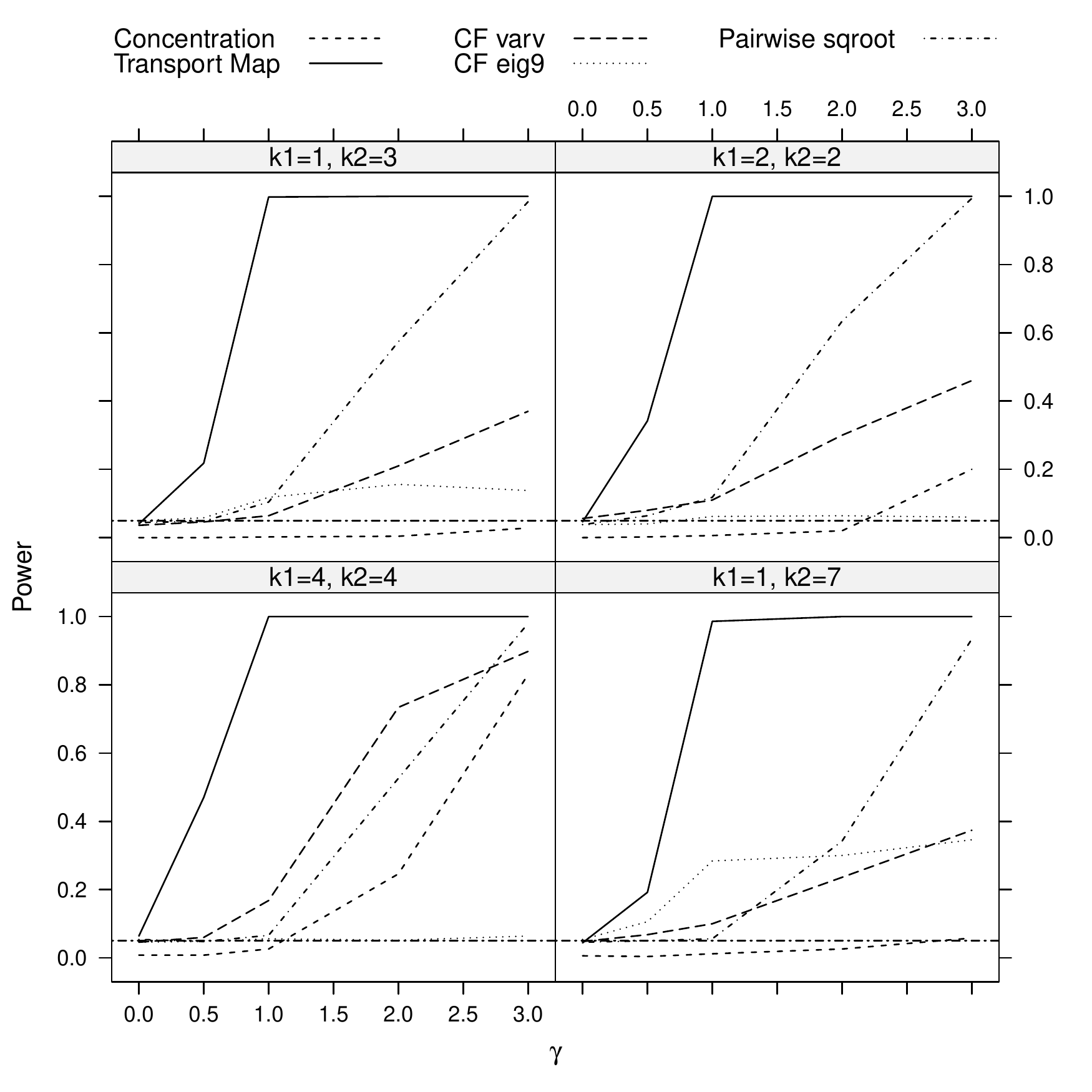}
\end{minipage}
  \begin{minipage}{0.49\linewidth}
\includegraphics[width=\textwidth]{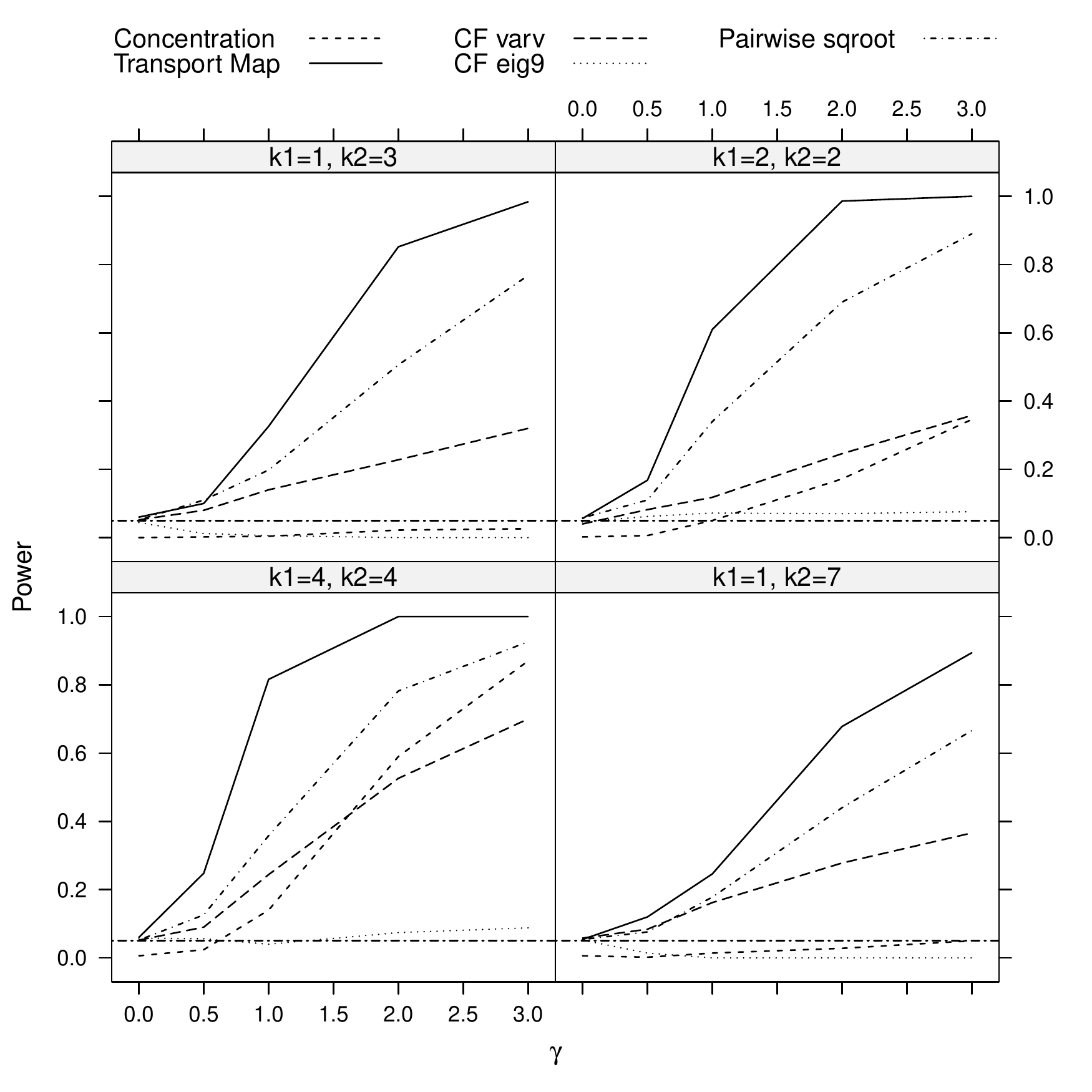}
\end{minipage}
\caption{Empirical power of ours and \citet{cabassi2017permutation, kashlak2019inference, hlavka2022functional}'s method in the Gaussian case, as a function of the perturbation parameter $\gamma$.  Left: geodesic perturbations;  right:  additive perturbations.}
\label{fig:power-Gaussian}
\end{figure}
\end{center}

\subsection{Tangent space PCA}\label{subsec-tgPCA}
We validate the PCA framework on a synthetic datasets which is inspired by the theoretical
generative model \eqref{eq:generativeModel}  and which yields $N$ covariances well separated in
$K$ groups. The aim is to see whether PCA is able to differentiate between
the groups. The operators $\Sigma_1, \dots, \Sigma_K$ are obtained as a conjugation perturbation of some known Fr\'echet mean by the generated ``optimal'' maps $T_1, \dots, T_K$
\[T_i= \sum_n \delta_n^{(i)} \sin(2n\pi t - \theta^{(i)}) \sin(2n\pi t - \theta^{(i)})\]
where the $\delta_n^{(j)}$ are drawn from a  $\chi^2$ distribution and $\theta^{(i)}$ are sampled from a von Mises distribution of mean $0$ and measure of concentration
$1/\sigma$.
The Fr\'echet mean is chosen to be
$
\bar\Sigma
=U \Lambda
U^*
$
as in \citet{kashlak2019inference}, with $U$ being a randomly
generated unitary operator, and $\Lambda$ a $d\times d$ diagonal
matrix with eigenvalue decay of $O(d^{-4})$, $d$ being the dimension
of the matrices.\\
As the generative model yields optimal maps
which are small perturbations of the identity, the dimension of
the matrices used to approximate the operators needs to be large,
otherwise the estimation errors would overwhelm the intrinsic
variability of the sample.
The dimension is chosen to be 200, the measure of concentration to be
1 and the number of groups $K$ to be $K=3$. For each of the $\Sigma_j$, $j=1,2,3$, we generate 100
samples of 50 Gaussian curves each. We then estimate the empirical
covariance of these curves, obtaining a sample of $N=300$ covariances. Results of the PCA are shown in
Table~\ref{importanceofcomponents_gen} and Figures~\ref{pcapairs_gen}. The Figures show that the different groups are clearly identified. 

\begin{table}
\centering
\begin{tabular}{rrrrrrrr}
  \hline
 & PC1 & PC2 & PC3 & PC4 & PC5 & PC6\\ 
  \hline
Standard deviation & 0.5810 & 0.1545 & 0.0966 & 0.0371 & 0.0276 & 0.0216 \\
  Proportion of Variance & 0.8989 & 0.0635 & 0.0248 & 0.0037 & 0.0020 & 0.0012 \\ 
  Cumulative Proportion & 0.8989 & 0.9624 & 0.9873 & 0.9909 & 0.9930 & 0.9942 \\ 
   \hline
\end{tabular}
\caption{Importance of each PC, first experiment with the
  generative model.}\label{importanceofcomponents_gen} 
\end{table}
\begin{center}
\begin{figure}
 \begin{minipage}{0.49\linewidth}
\includegraphics[width=\textwidth]{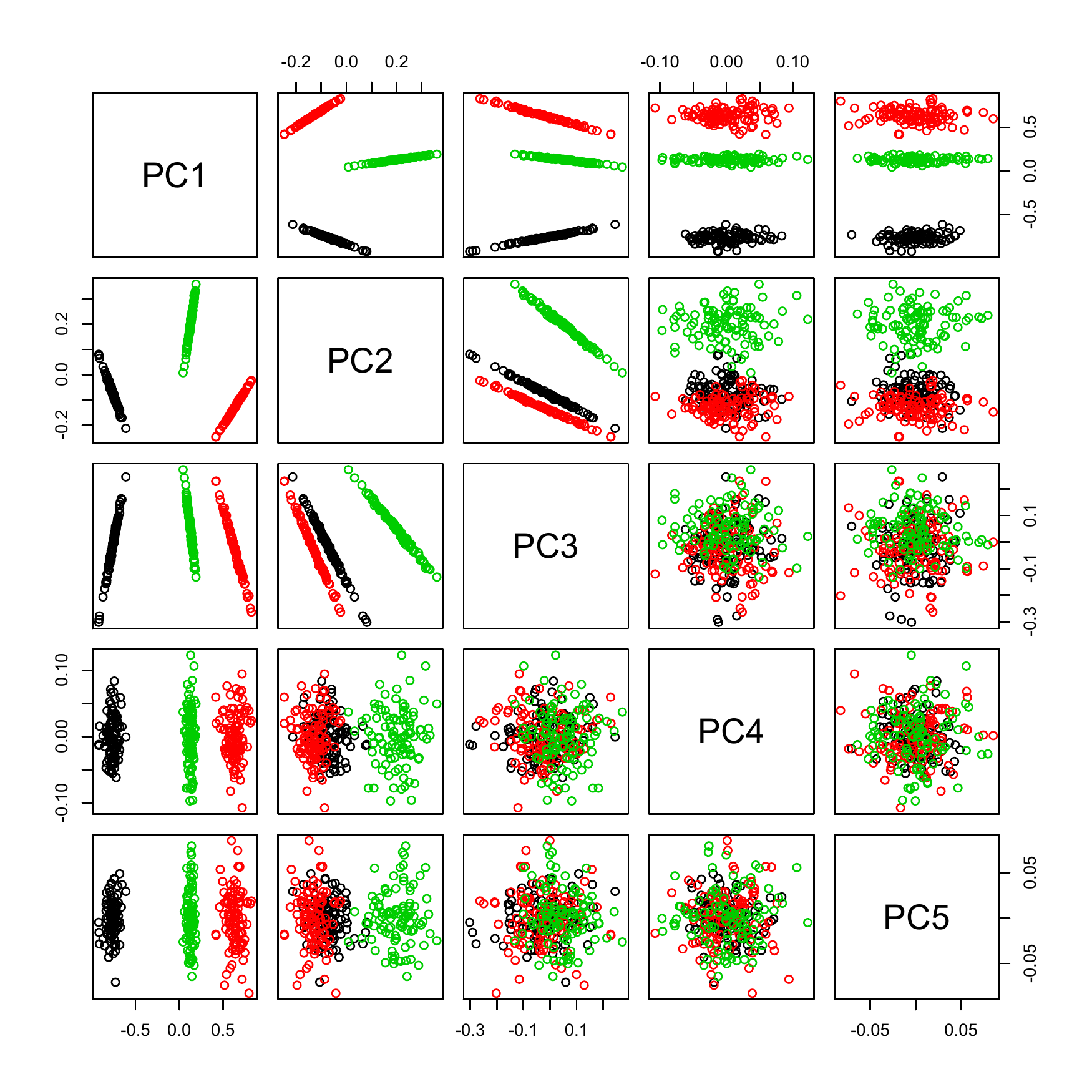}
\end{minipage}
\begin{minipage}{0.49\linewidth}
\includegraphics[width=\textwidth]{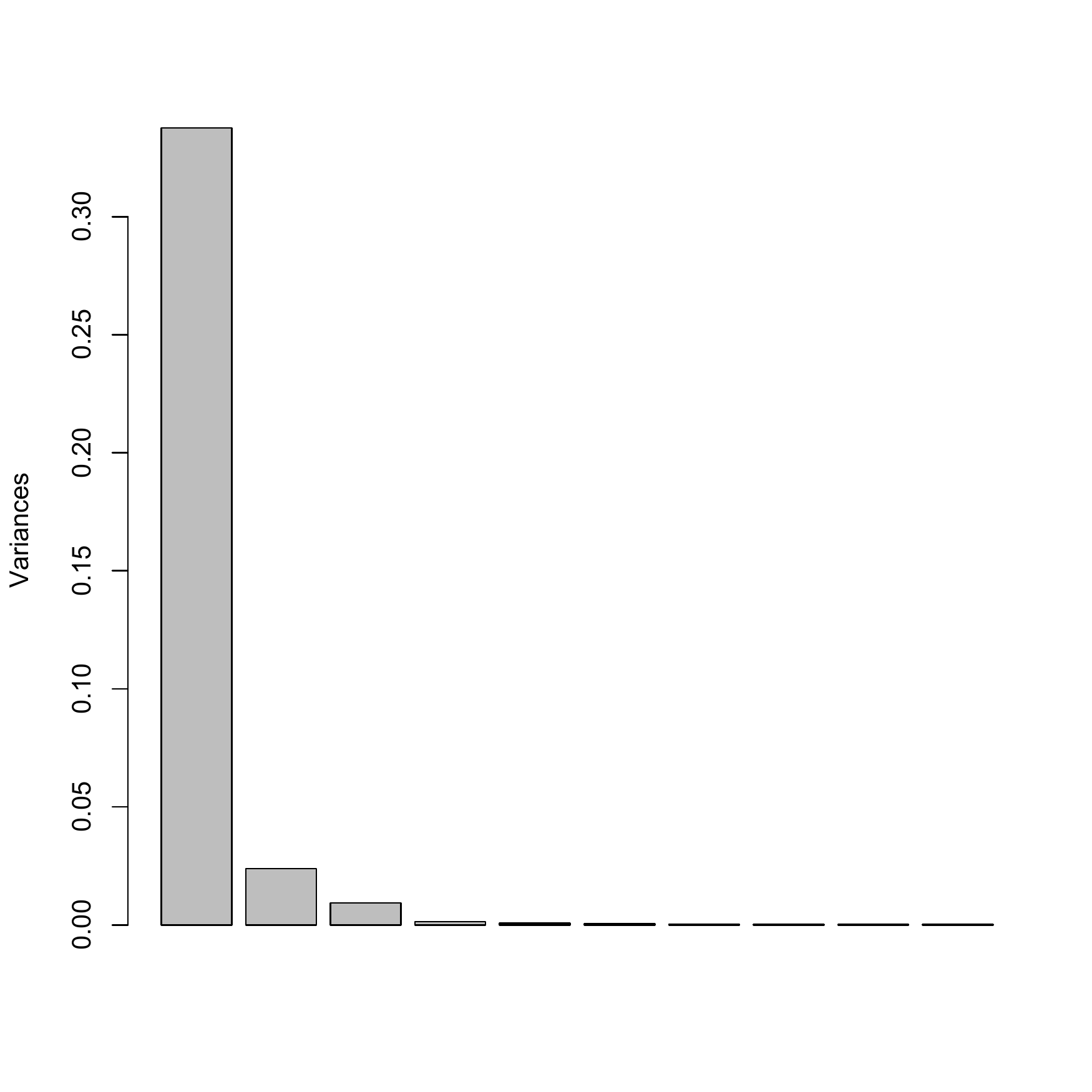}
\end{minipage}
\caption{Left: PCA scores, first experiment with the generative
  model. Colours correspond to the three maps generated from the model. Right: eigenvalues screeplot, first experiment with the generative model}\label{pcapairs_gen}
\end{figure}
\end{center}

\section{Data Analysis:  Phoneme Periodograms}\label{sec-dataanalysis}
In this section, we illustrate our method on the phoneme data set considered in \citet{hastie1995penalized}. The dataset consists of 4509 log-periodograms of length 256 each, computed from continuous speech frames of 50 male speakers. Each speech frame is 32msec long, sampled at a rate of 16kHz and represents one of the following five phoneme: ``aa'' (as in ``dark'', nasal a), ``ao'' (as in  ``water''), ``iy'' (as in ``she''), ``sh'' (as in ``she''), ``dcl'' (as in ``dark'', ``british'' d).  Each phoneme $j$ gives rise to a covariance operator $\Sigma_j$.
We use this sample of $K=5$ covariances to generate $K$ populations of $n$ Gaussian processes, on which inference will be performed.

\subsection{ANOVA}\label{subsec:anova-phoneme}
In order to perform ANOVA on the phoneme dataset, we extract the
log-periodograms corresponding to the phonemes ``aa'', ``ao'', ``iy'' . We limit the test to these three phonemes because of their similarity, which makes it harder to distinguish them and allows for better discrimination among different procedures. 
If we include all 5 phonemes, the difference between vowels and consonants sound is so stark that all tests have very high
power.

To sample under $H_0$, we sample $3n$ log-periodograms for the "iy"
phoneme which are then randomly
assigned to three groups, each of size $n$. To sample under the
alternative $H_1$, we sample $n$ log-periodograms for each phoneme.
We repeat the test for $n=25$ and $n=50$, and for 500 replications and
200 permutations. 
Again we compare both with \citet{cabassi2017permutation},
\citet{kashlak2019inference} and \citet{hlavka2022functional}. We limit the comparison with \citet{hlavka2022functional} to the test statistics using the sample covariance matrix, as it is the scenario that gives consistently better results. We extend their two-sample testing to 3-samples by considering all pairwise comparisons and using the maximum test statistics. When interpreting the results, it is important to treat the outputs
carefully, since the procedure of \citet{kashlak2019inference} was unable
to produce a result in a small number of cases, as the computation of the
distance using SVD failed, while in some cases of \citet{hlavka2022functional}, both the test statistics under the null and the p-values are identically 0. 

Table \ref{tab:ktest-phonema} shows the comparison between the test on smoothed phonema log-periodograms.  Since the different phonemes have different mean functions, observations must be
centered around the sample mean of each group, implying that the right type I error probability under
$H_0$ might not be respected. Regardless, the transport test delivers a level very close to
the nominal 0.05, especially when $n=50$ (which is still
relatively low compared to the 256 points where the curves are
sampled). The tests of \citet{cabassi2017permutation} and \citet{hlavka2022functional} also reache an acceptable nominal level under $H_0$. This is not the case for \citet{kashlak2019inference} making impossible the comparison. However, \citet{cabassi2017permutation} show that their test outperforms that of \citet{kashlak2019inference}.  
It is worth mentioning that the test of \citet{hlavka2022functional} responds very well to variation of the mean, because it takes into account the full distribution thanks to the characteristic functionals. In the phoneme dataset, the mean log-periodogram captures most of the variability. Thus, if one was to consider variation with respect to both first and second moment simultaneously, the method of \citet{hlavka2022functional} would show a significant increase in power.  However, for the scope of this work, we are interested specifically in second order variation. When centering with respect to the mean, the test based on the transport maps greatly outperforms the competing methods under the alternative
hypothesis.

\begin{table}[t]
\centering
\begin{tabular}{rrrrrr}
  \hline
& $n$ & Pairwise & Concentration & Transport Maps & ECF-varv\\ 
  \hline
 $H_0$& 25 & 0.086 & 0.000 & 0.068 & 0.2\\ 
  &  50 & 0.050 & 0.000 & 0.046 & 0.5\\
 $H_1$& 25 & 0.271 & 0.300 & 0.470 & 0.01\\
  &  50 & 0.670 & 0.944 & 0.994 & 0.1\\
   \hline
\end{tabular}
\caption{Comparison of the empirical power of the three different testing methods on the
  \texttt{phoneme} dataset, when applied on the phonemes "aa", "ao" and "iy".
}\label{tab:ktest-phonema}
\end{table}

\subsection{PCA}\label{subsec:pca-phoneme}
In this section, we illustrate the use of tangent space PCA of covariance operators by applying it to the phoneme dataset described in \citet{hastie1995penalized}. The collection of curves corresponding to each phoneme gives rise to a sample covariance operator, for a total of five covariances. The five empirical covariances are lifted to the tangent space via the log map centered at their Fr\'echet mean $\bar\Sigma$. Successively they are scaled by $\bar\Sigma^{1/2}$ as explained in section \ref{subsec:transport_pca}. Standard PCA can now be run on these quantities. 
Figure \ref{fig:PCA-phonemas} (left) shows the results of applying PCA on \texttt{phoneme} data. We see clearly that tangent space PCA captures very well the difference among the phonemes, as each sillable is isolated in at least one plot. The colours are as follows: ``sh'' black, ``iy''
  red, ``dcl ``green'', ``aa'' blue, ``ao'' cyan, more precisely:
 \begin{enumerate}
  \item The first PC captures (part of) the difference between ``aa, ao and iy" (vowels) and ``dcl and sh" (consonants)
  \item   The second PC captures (part of) the difference between ``dcl" and ``sh" (two consonants).
  \item   The third PC captures (part of) the difference between ``aa and ao" and ``iy" (separating the two similar sounding vowels from the third more different one).
  \item   The fourth PC captures (part of) the difference between ``aa" and ``ao" (separating the last two remaining, and very similar, sounds).
  \item    Since, the order in the $y$-axes is the order of magnitude of the eigenvalues the analysis suggests also the importance of the differences between the operators. As intuition dictates the difference between vowels and consonants is nearly four times more pronounced than the difference between the sounds ``aa" and ``ao".
\end{enumerate}
The screeplot (Figure~\ref{fig:PCA-phonemas}, right) shows that four PCs explain the full variance of the data, which is obvious as we have only five data points.  The fourth PC is quite important and explains 13\% of the variance.

In order to test our methodology in a more realistic situation, we artificially enlarge our sample of covariances by subsampling the original data.  Specifically, from each of the five phonemes we subsample $B=50$ of the corresponding log-periodograms to obtain a new estimator of the covariance operator of that phoneme.  We do this $G=12$ times so that in total we have a sample of $5G=60$ covariances, divided into five groups, and the covariances in a group should be close to each other.  We then carry out the PCA on these 60 covariances
 The results are showed in Figure~\ref{fig:PCA-phonemas} (right). Again, we can see that each phonema is isolated in at least one plot. Figure \ref{fig:pcacomparison} shows the comparison of the PCA scores both in Euclidean and in Wasserstein distance.  It is seen that the PCA based on the Procrustes tangent space distance is much more successful in distinguishing the covariances of different phonemes.

\begin{center}
  \begin{figure}
    \begin{minipage}{0.49\linewidth}
\includegraphics[width=\textwidth]{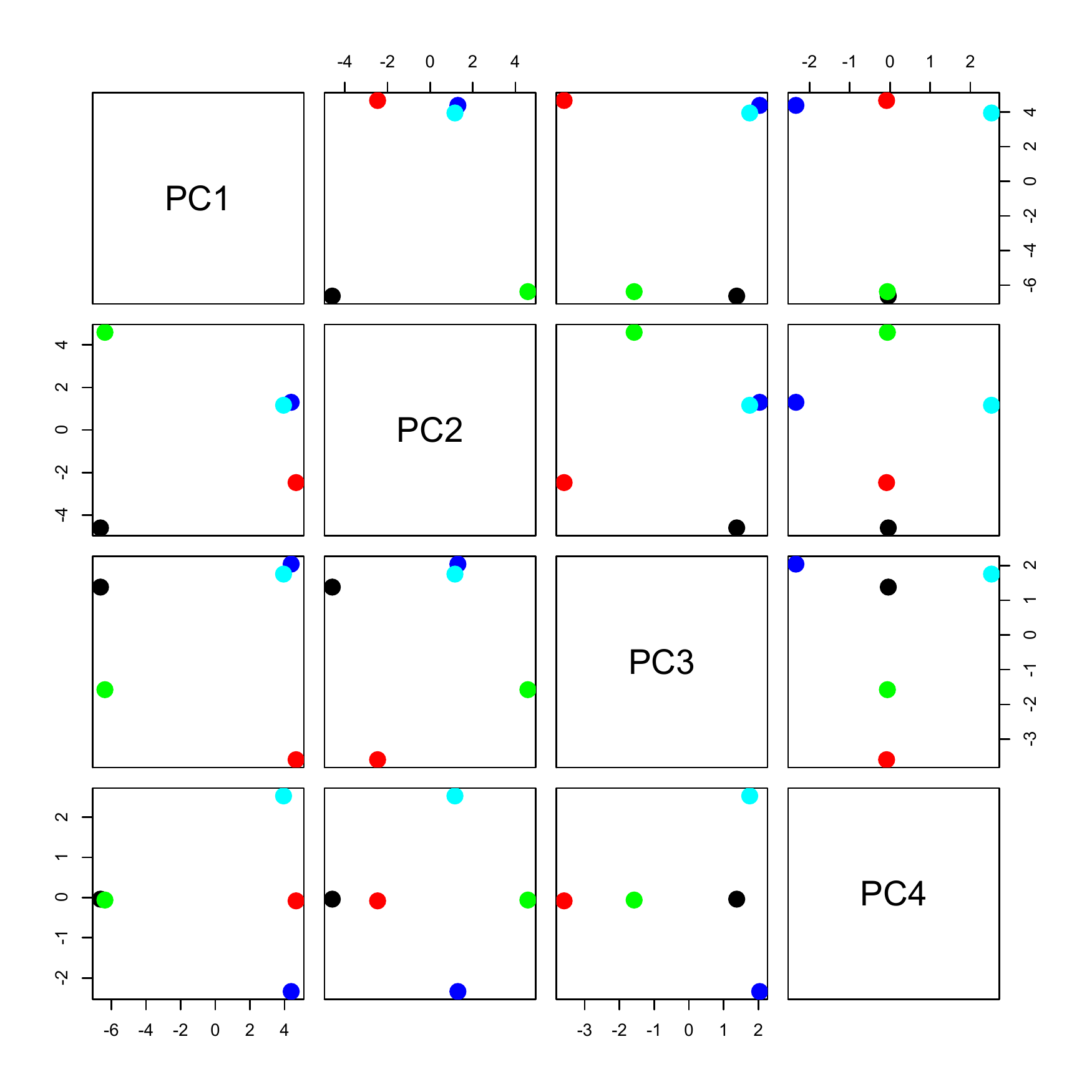}
\end{minipage}
\begin{minipage}{0.49\linewidth}
  \includegraphics[width=\textwidth]{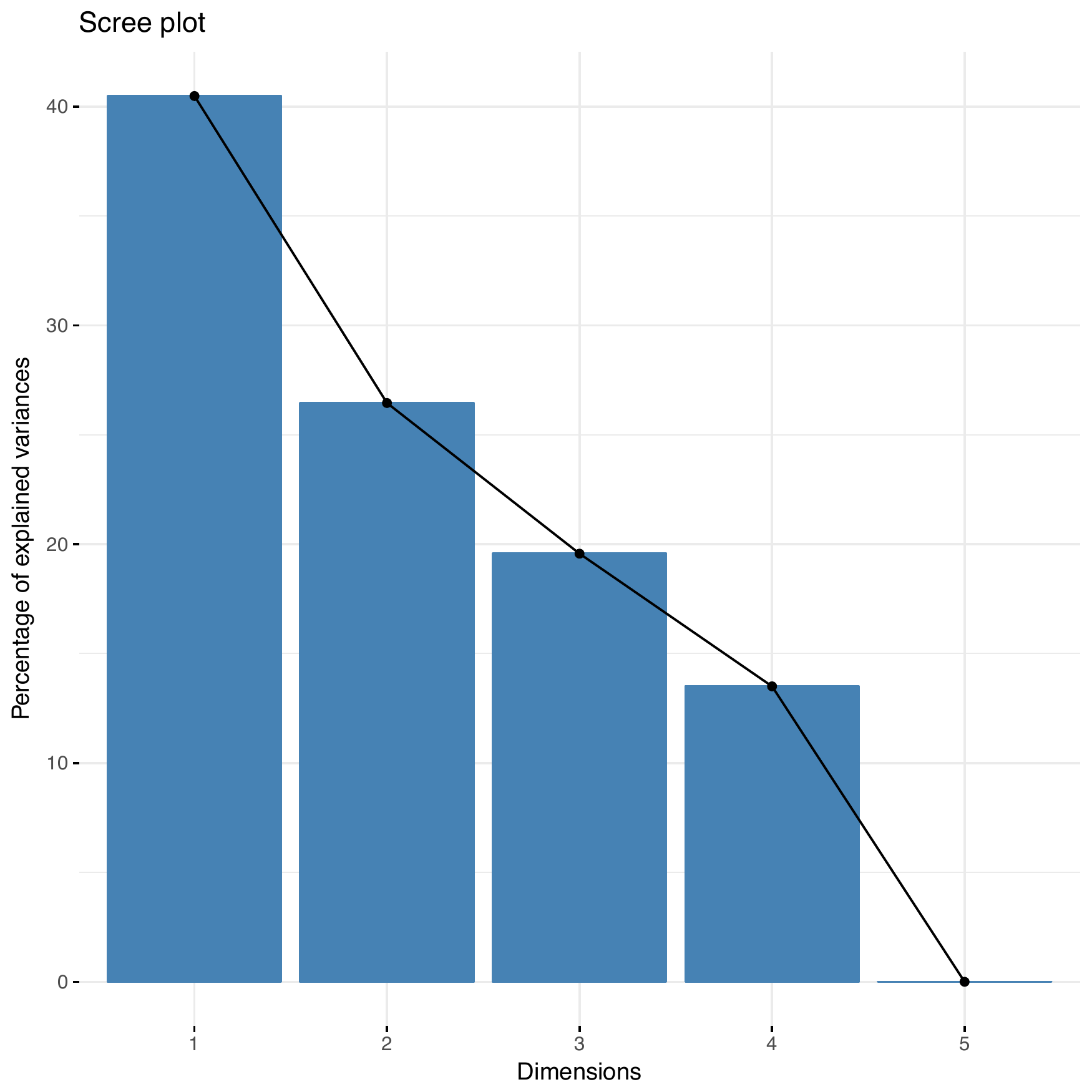}
  \end{minipage}
\caption{Left:  PCA scores, as computed from the phoneme dataset. The colours are as follows: ``sh'' black, ``iy''
  red, ``dcl ``green'', ``aa'' blue, ``ao'' cyan.  Right: screeplot of eigenvalues, phoneme dataset}
\label{fig:PCA-phonemas}
\end{figure}
\end{center}
\begin{center}
  
  \begin{figure}
    \begin{minipage}{0.49\linewidth}
      \includegraphics[width=\textwidth]{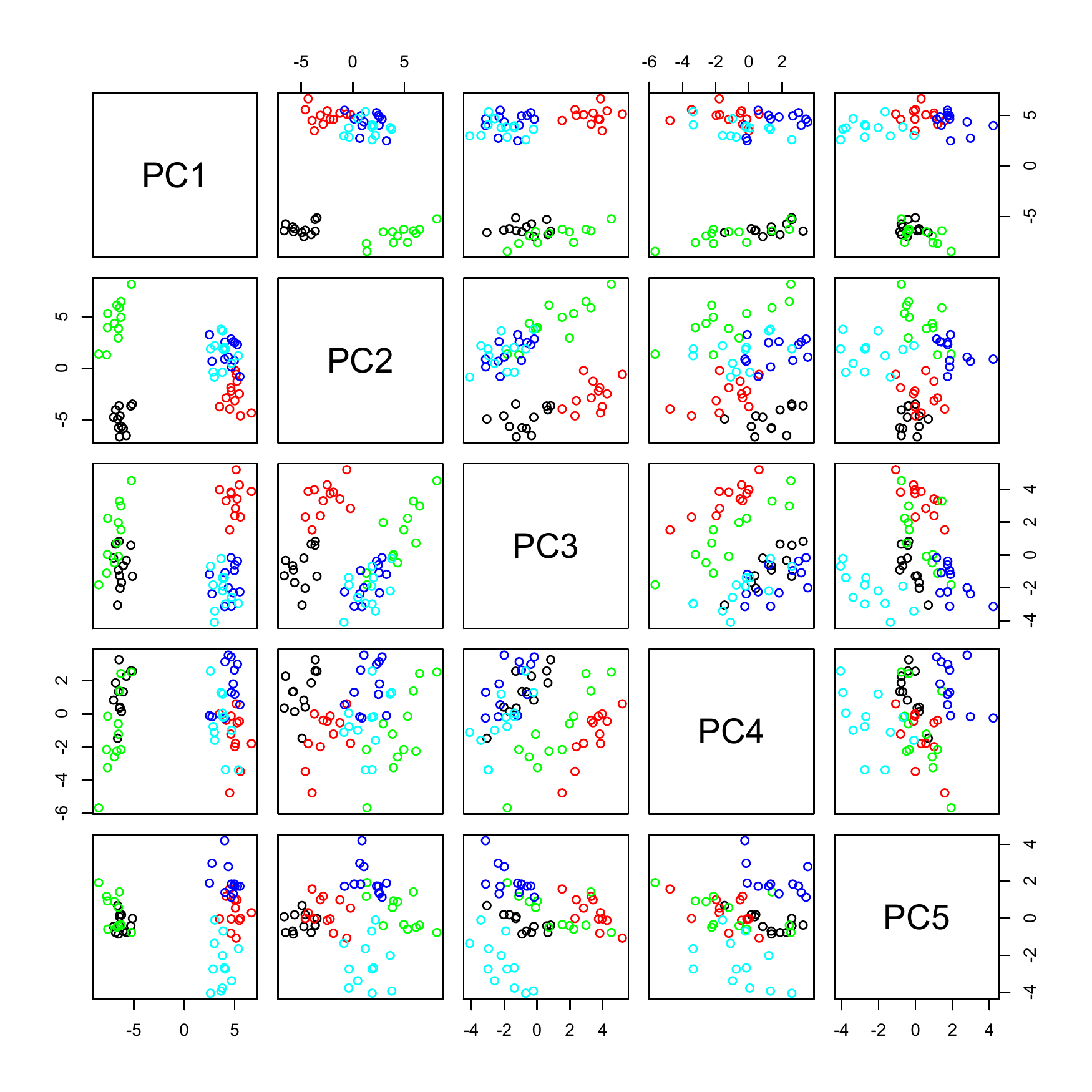}
\end{minipage}
\begin{minipage}{0.49\linewidth}
\includegraphics[width=\textwidth]{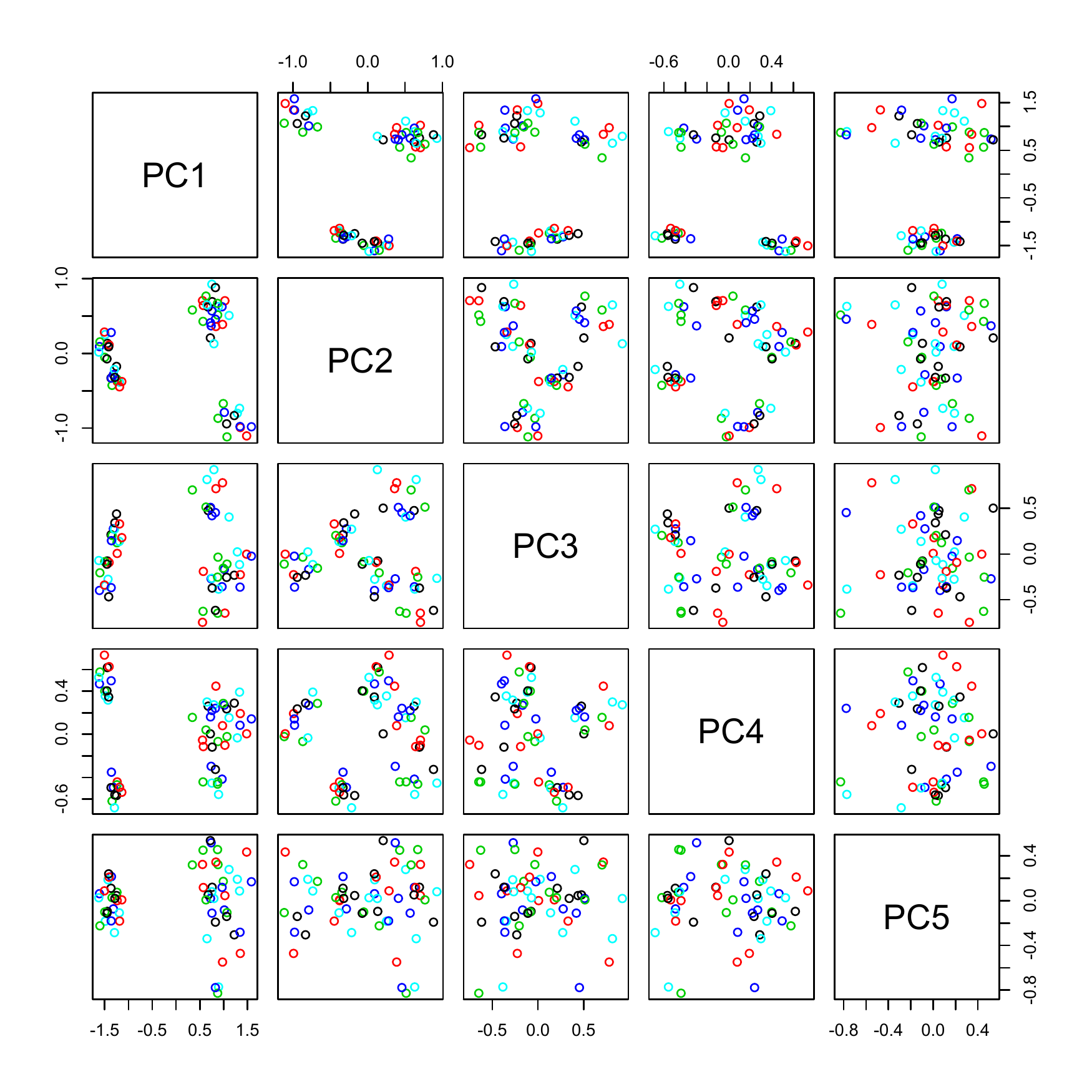}
\end{minipage}
\caption{PCA of the 60 covariance operators, based on the Procrustes tangent space distance (left) and the Hilbert--Schmidt distance (right)}
\label{fig:pcacomparison}
\end{figure}
\end{center}

\section{Concluding Remarks }\label{sec:conclusions}
This paper introduces a framework that allows the comparison of several population of stochastic processes with respect to their covariance structure. 
We contributed a new methodology that exploits the theory of optimal (multi)transport and demonstrate how taking such a stand point allows to develop: (a) a testing procedure which outperforms the state of the art and (b) the first instance of tangent space principal component analysis of covariance operators. A fundamental ingredient of our approach and the main theoretical contribution of this paper is the proof that Gaussian measures can always be multicoupled through bounded non-negative linear operators. The existence and boundedness result is elemental to both the testing procedure, and the PCA: the test statistic compares these coupling maps to the identity in norms of various strenghts, whereas the PCA uses the deviations of these couplings from the identity as a basis for eigenanalysis. 

Specifically concerning the testing procedure, an appealing aspect of the presented methodology is that it harnesses a genuinely functional effect, in order to manifest exceptionally powerful performance under wide classes of alternatives -- -- alternatives for which previous tests would not perform nearly as well.

The genuinely functional effect arises under alternatives where the optimal coupling maps are \emph{not} Hilbert--Schmidt (they are merely guaranteed to be bounded). At the population level, this corresponds to $\hs T_j - \mathscr I\hs_2=\infty$, yielding very large values of the empirical test statistic. 

Such a situation arises when departures from the null happen ``across the whole spectrum" and not just in its bulk. These situations are not a theoretical curiosity ---they can indeed be very common, as the following example will illustrate.

Consider the two-sample setting, and suppose that $\Sigma_1$ is the covariance operator of standard Brownian motion on $[0, 1]$, whereas $\Sigma_2=\sigma^2 \Sigma_1$ is the covariance of standard Brownian motion on [0,1] scaled by the positive scalar $\sigma>0$. These two covariance operators commute with Fr\'echet mean having covariance $(1+\sigma)^2\Sigma_1/4$, and corresponding transport maps $\mathbf t_1=[2\sigma/(1+\sigma])]\mathscr I$ and $\mathbf t_2=[2/(1+\sigma])]\mathscr I$. Thus $\Delta_j$ are both bounded, but they are not Hilbert--Schmidt unless $\sigma=1$, leading to $\hs{\Delta_2}\hs_k=\infty$, for any $k\geq 1$. The Wasserstein distance itself, however, equals $\hs \Sigma_1^{1/2}-\Sigma_2^{1/2} \hs_2=(1-\sigma)\hs \Sigma_1^{1/2}\hs_2= (1-\sigma)/\sqrt{12}$ and thus becomes arbitrarily small as $\sigma$ nears one (recall that in the commutative case, the Wasserstein distance becomes the Hilbert--Schmidt distance of the corresponding positive roots). 

The functional effect taking place is related to the Hajek--Feldman alternative, which has been exploited to obtain perfect discrimination of Gaussian process differing in their mean within an FDA context (see \citet{delaigle2012achieving}), by similarly exploiting the fact that a certain norm diverges under the alternative. 

To see the connection with our setting, assume $\Sigma_j=\sum_n \lambda_{j,n}\varphi_{n}\otimes\varphi_n$ ($j=1,2$) have the same eigenfunctions and thus commute.  Then zero-mean Gaussian measures $N(0,\Sigma_1)$ and $N(0,\Sigma_2)$ are equivalent if and only if $\sum_n(r_n - 1)^2$ converges, where $r_n=\lambda_{2,n}/\lambda_{1,n}$.  Indeed, summability implies that $r_n\to1$ so that $\Sigma_1^{1/2}$ and $\Sigma_2^{1/2}$ have the same range and
\[
\hs(\Sigma_1^{-1/2}\Sigma_2^{1/2})(\Sigma_1^{-1/2}\Sigma_2^{1/2})^*-\mathscr{I}\hs_2 = \sum_{n=1}^\infty (r_n - 1)^2
\]
is finite, as required (see e.g., \citet[Theorem~2.25]{da2014stochastic}).  The Fr\'echet mean and transport maps are
\[
\overline \Sigma = \sum_{n=1}^\infty\left[\frac{\sqrt{\lambda_{1,n}} + \sqrt{\lambda_{1,n}}}2\right]^2 \varphi_n\otimes\varphi_n
,\quad \mathbf t_1
=\frac2{1+\sqrt{r_n}} \varphi_n\otimes\varphi_n,
\quad 
\mathbf t_2
=\frac2{1+\sqrt{r_n^{-1}}}\varphi_n\otimes\varphi_n,
\]
and simple algebra shows
\[
\hs\Delta_1\hs_2<\infty \iff \sum_{n=1}^\infty (r_n - 1)^2 \iff \hs\Delta_2\hs_2<\infty.
\]
Thus, in the commutative case, the population level test statistic is finite if and only if the Gaussian measures are equivalent.  Whether or not this is the case depends on how differences persist across the whole spectrum, rather than just in the bulk.

\section{Proofs of Formal Statements}\label{sec:proofs}
\begin{proof}[Proof of Lemma~\ref{lem:nullHyp}]
If $\Sigma_1=\cdots=\Sigma_K$, then the unique optimal multicoupling is given by the maps $\mathbf{t}_j(z)=z$ and the process $Z\sim N(0,\Sigma_1)$.  Conversely, if a multicoupling of $(\gamma_1,\ldots,\gamma_K)$ is achieved as the law of $(\mathbf{t}_1(Z),\ldots,\mathbf{t}_K(Z))$ for some process $Z$ and some maps satisfying $\mathbf{t}_1=\cdots=\mathbf{t}_K$, then $\gamma_i$, the law of $\mathbf{t}_i(Z)$, is the same for all $i$, i.e., $\gamma_1=\cdots=\gamma_K$, and so $\Sigma_1=\cdots=\Sigma_K$.  Uniqueness of the optimal multicoupling follows from the first sentence in the proof.
\end{proof}
For the proof of Theorem~\ref{thm:bounded_maps}, we need the following result from \citet{douglas1966majorization}.
\begin{lemma}\label{lem:Baker}
Let $0\le A\le B$ be bounded operators, where $A\le B$ means that $B-A$ is non-negative.  Then there exists a bounded operator $G$ with $\hs G\hs_\infty\le 1$ such that $A^{1/2}=B^{1/2}G$ and $\ker G^*\supseteq \ker B$.
\end{lemma}

\begin{proof}[Proof of Theorem~\ref{thm:bounded_maps}]
Let $\overline\Sigma$ be any Fr\'echet mean of $\Sigma_1,\dots,\Sigma_K$ and define $Q_i=(\Sigma^{1/2}\Sigma_i\Sigma^{1/2})^{1/2}\ge 0$.  The fixed point equation for Fr\'echet means (\cite[Proposition~16]{masarotto2018procrustes}) yields the inequality
\[
Q_i
\le \sum_{j=1}^K Q_j
=K\overline\Sigma.
\]
By Lemma~\ref{lem:Baker} there exists an operator $G$ with range included in the closed range of $\overline\Sigma$ such that $Q_i^{1/2}=\overline\Sigma^{1/2}G$, $\hs G\hs_\infty \le\sqrt K$ and we may write $G=\overline\Sigma^{-1/2}Q_i^{1/2}$.  If we can identify $G^*$ with $Q_i^{1/2}\overline\Sigma^{-1/2}$, then we can conclude that
\[
\topt i{}
=\topt{\overline\Sigma}{\Sigma_i}
=\overline\Sigma^{-1/2}(\overline\Sigma^{1/2}\Sigma_i\Sigma^{1/2})^{1/2}\overline\Sigma^{-1/2}
=GG^*
\]
is well-defined and bounded, with operator norm bounded by $K$.  Let $y=\overline\Sigma^{1/2}z$ and notice that for all $x$
\[
\innprod{G^*y}x
=\innprod{\overline{\Sigma}^{1/2}z}{Gx}
=\innprod{\overline{\Sigma}^{1/2}z}{\overline{\Sigma}^{-1/2}Q_i^{1/2}x}
=\innprod{z}{Q_i^{1/2}x}
=\innprod{Q_i^{1/2}\overline\Sigma^{-1/2}y}x.
\]
Thus $G^*=Q_i^{1/2}\overline\Sigma^{-1/2}$ on $\range(\overline\Sigma^{1/2})$.  Since $G^*$ is bounded the equality extends to the closure of the range, which is $(\ker\overline\Sigma)^\perp$.  On $\ker\overline\Sigma$ both operators are identically zero.

We have thus established the existence of deterministic optimal maps from the Fr\'echet mean $\overline\Sigma$ to each of the operators $\Sigma_j$.  Now if $Z\sim N(0,\overline\Sigma)$, then $\pi=(\topt1{}(Z),\dots,\topt K{}(Z))$ is a multicoupling of the corresponding Gaussian measures, and the optimality of $\pi$ follows from \citet[Proposition~2]{zemel2019frechet}.
\end{proof}

\bibliographystyle{imsart-nameyear}
\bibliography{fda}

\end{document}